\documentclass[11pt,onecolumn]{IEEEtran}

\usepackage{amsfonts}
\usepackage{times,amssymb,amsthm,amsmath,amsfonts,bbm,mathrsfs,color,subfigure,graphicx,enumitem,booktabs,cite,euscript,tcolorbox}
\usepackage[font=footnotesize]{caption}
\allowdisplaybreaks
\usepackage[unicode]{hyperref}
\hypersetup{
	colorlinks,
	linkcolor={blue!80!black},
	citecolor={blue!80!black},
	urlcolor={blue!80!black}
}

\newtheorem{theorem}{Theorem}[section]
\newtheorem{lemma}[theorem]{Lemma}

\newtheorem{proposition}[theorem]{Proposition}

\newtheorem{corollary}[theorem]{Corollary}

\newtheorem{construction}{Construction}[section]
\theoremstyle{remark}
\newtheorem{definition}{Definition}[section]
\newtheorem{example}{Example}[section]
\newtheorem{remark}{Remark}
\usepackage[top=1in,bottom=1in,left=1in,right=1in]{geometry}

\newcommand\nc\newcommand
\nc\ffa{{\boldsymbol a}}\nc\ffA{{\boldsymbol A}}\nc\cA{{\EuScript A}}
\nc\ffb{{\boldsymbol b}}\nc\ffB{{\boldsymbol B}}\nc\cB{{\EuScript B}}\nc\bG{{\mathbb G}}
\nc\ffc{{\boldsymbol c}}\nc\ffC{{\boldsymbol C}}\nc\cC{{\mathscr C}}
\nc\ffd{{\boldsymbol d}}\nc\ffD{{\boldsymbol D}}\nc\cD{{\EuScript D}}
\nc\ffe{{\boldsymbol e}}\nc\ffE{{\boldsymbol E}}\nc\cE{{\EuScript E}}
\nc\fff{{\boldsymbol f}}\nc\ffF{{\boldsymbol F}}\nc\cF{{\mathscr F}}\nc\sF{{\EuScript F}}
\nc\ffg{{\boldsymbol g}}\nc\ffG{{\boldsymbol G}}\nc\cG{{\EuScript G}}
\nc\ffh{{\boldsymbol h}}\nc\ffH{{\boldsymbol H}}\nc\cH{{\EuScript H}}
\nc\ffi{{\boldsymbol i}}\nc\ffI{{\boldsymbol I}}\nc\cI{{\mathcal I}}
\nc\ffj{{\boldsymbol j}}\nc\ffJ{{\boldsymbol J}}\nc\cJ{{\EuScript J}}
\nc\ffk{{\boldsymbol k}}\nc\ffK{{\boldsymbol K}}\nc\cK{{\EuScript K}}
\nc\ffl{{\boldsymbol l}}\nc\ffL{{\boldsymbol L}}\nc\cL{{\EuScript L}}
\nc\ffm{{\boldsymbol m}}\nc\ffM{{\boldsymbol M}}\nc\cM{{\EuScript M}}
\nc\ffn{{\boldsymbol n}}\nc\ffN{{\boldsymbol N}}\nc\cN{{\EuScript N}}
\nc\ffo{{\boldsymbol o}}\nc\ffO{{\boldsymbol O}}\nc\cO{{\EuScript O}}
\nc\ffp{{\boldsymbol p}}\nc\ffP{{\boldsymbol P}}\nc\cP{{\EuScript P}}
\nc\ffq{{\boldsymbol q}}\nc\ffQ{{\boldsymbol Q}}\nc\cQ{{\EuScript Q}}
\nc\ffr{{\boldsymbol r}}\nc\ffR{{\boldsymbol R}}\nc\cR{{\EuScript R}}
\nc\ffs{{\boldsymbol s}}\nc\ffS{{\boldsymbol S}}\nc\cS{{\EuScript S}}
\nc\fft{{\boldsymbol t}}\nc\ffT{{\boldsymbol T}}\nc\cT{{\EuScript T}}
\nc\ffu{{\boldsymbol u}}\nc\ffU{{\boldsymbol U}}\nc\cU{{\EuScript U}}
\nc\ffv{{\boldsymbol v}}\nc\ffV{{\boldsymbol V}}\nc\cV{{\mathscr V}}
\nc\ffw{{\boldsymbol w}}\nc\ffW{{\boldsymbol W}}\nc\cW{{\mathscr W}}
\nc\ffx{{\boldsymbol x}}\nc\ffX{{\boldsymbol X}}\nc\cX{{\EuScript X}}
\nc\ffy{{\boldsymbol y}}\nc\ffY{{\boldsymbol Y}}\nc\cY{{\mathscr Y}}
\nc\ffz{{\boldsymbol z}}\nc\ffZ{{\boldsymbol Z}}\nc\cZ{{\EuScript Z}}
\nc{\bb}{{\mathbbm{1}}}
\nc\reals{{\mathbb R}}
\nc{\ff}{{\mathbb F}}
\nc{\PP}{{\mathbb P}}

\DeclareMathOperator{\diam}{{\sf diam}}

\DeclareMathOperator{\sgn}{sgn}
\DeclareMathOperator{\rank}{rk}

\DeclareMathOperator*{\argmin}{arg\,min}

\usepackage{tikz,tikz-cd}
\usetikzlibrary{positioning}
\usetikzlibrary{calc}

\usetikzlibrary{matrix,arrows,decorations.pathmorphing}
\usetikzlibrary{decorations.pathreplacing}

\DeclareSymbolFont{bbold}{U}{bbold}{m}{n}
\DeclareSymbolFontAlphabet{\mathbbold}{bbold}

\nc{\remove}[1]{}

\usepackage{xargs}
\usepackage[colorinlistoftodos,prependcaption,textwidth=.7in]{todonotes}
\newcommandx{\yellownote}[2][1=]{\todo[linecolor=yellow,backgroundcolor=yellow!25,bordercolor=yellow,#1]{#2}}
\newcommandx{\unsure}[2][1=]{\todo[linecolor=red,backgroundcolor=red!25,bordercolor=red,#1]{#2}}
\newcommandx{\change}[2][1=]{\todo[linecolor=blue,backgroundcolor=orange!25,bordercolor=blue,#1]{#2}}
\newcommandx{\info}[2][1=]{\todo[linecolor=OliveGreen,backgroundcolor=OliveGreen!25,bordercolor=OliveGreen,#1]{#2}}
\newcommandx{\greennote}[2][1=]{\todo[linecolor=olive,backgroundcolor=green!25,bordercolor=olive,#1]{#2}}
\usepackage[normalem]{ulem}
\newcommand\redout{\bgroup\markoverwith{\textcolor{red}{\rule[0.5ex]{2pt}{0.8pt}}}\ULon}

\nc\nnfootnote[1]{%
   \begin{NoHyper}
    \footnote{#1}%
    \addtocounter{footnote}{-1}%
   \end{NoHyper}
}

\begin{document}
	
	\title{Generalized regenerating codes and node repair on graphs}
	\author{\IEEEauthorblockN{Adway Patra} \hspace*{1in}
		\and \IEEEauthorblockN{Alexander Barg}}
	\date{}
	\maketitle
	
	\begin{abstract}
We consider regenerating codes in distributed storage systems where connections between the nodes are constrained by a 
graph. In this problem, the failed node downloads the information stored at a subset of vertices of the graph for the 
purpose of recovering the lost data. Compared to the standard setting, regenerating codes on graphs address two additional features. The repair information is moved across the network, and the cost of node repair is determined by the graphical distance from the helper nodes to the failed node. Accordingly, the helpers far away from the failed
node may be expected to contribute less data for repair than the nodes in the neighborhood of that node. 
We analyze regenerating codes with nonuniform download for repair on graphs. Moreover, in the process of repair, the information moved from the helpers to the failed node may be combined through intermediate processing, reducing the repair bandwidth. We derive lower bounds for communication complexity of node repair on graphs, including repair schemes with nonuniform download and intermediate processing, and construct codes that attain these bounds.

Additionally, some of the nodes may act as adversaries, introducing errors into the data moved in the network. 
 For repair on graphs in the presence of adversarial nodes, we construct codes that support node repair and error correction in systematic nodes.
	\end{abstract}

 \thispagestyle{empty}
			
			
			

\nnfootnote{\vspace*{-.15in}			
 \hspace*{-.25in}	\noindent\rule{1.5in}{.4pt}
   
   \vspace*{.15in}
 The results of this paper were presented in part at ISIT 2022 and ISIT 2023.
			
		The authors are with Dept. of ECE and ISR, University of Maryland, College Park, MD 20742. Emails: 
  \{apatra,abarg\}@umd.edu. 
 		This research was supported by NSF grants CCF2110113 (NSF-BSF), CCF2104489, and CCF2330909.} 
	
	\section{Introduction}\label{Introduction}
	Regenerating codes are designed to correct single or multiple erasures from the information obtained from the 
	nonerased coordinates of the codeword. Using the terminology inspired by distributed storage systems, a 
	regenerating code recovers a failed node (coordinate) by downloading data from the surviving nodes of the encoding. 
	Suppose that the code used to protect the data has length $n$ and that every coordinate is placed on a separate 
	storage node. To repair a failed node, the system uses a subset of $d\le n-1$ {\em helper nodes} that transmit 
	information comprising some functions of their contents to be used to recover the lost data. One of the goals of 
	the code design is to minimize the amount of data downloaded to complete the repair. Following their introduction 
	in \cite{Dimakis10}, regenerating codes have been the subject of voluminous literature devoted both to constructions 
	and impossibility bounds for the code parameters; see \cite{Ramkumar2022} for a very readable and detailed overview of
	this area.
	
	A distributed storage system is formed of a number of nodes connected by communication links which carry information to accomplish the two basic tasks performed in the system, namely data recovery and node repair. The amount of information sent over the links is a key metric of the system efficiency. The problem of node repair has been widely studied in the literature in the last decade.
	The system is modeled as $n$ storage nodes each with capacity of $l$ units, used to store a file $\cF$ of size $M$, such that the following two properties are met \cite{Dimakis10}:
	\begin{itemize}
		\item{(\em{Data retrieval})} The entire file can be recovered by accessing any $k < n$ nodes. 
		\item{(\em{Repair})} If a single node fails, data from any $d$ surviving, or helper nodes is used to restore the lost data. 
		We assume that each of the helper nodes contributes $\beta\le l$ units of data, and that $k \le d \le n - 1$. The parameter $\beta$ is 
		called the {\em per-node repair bandwidth}.
	\end{itemize}
	The fundamental tradeoff between the file size and the repair bandwidth is expressed by the bound of \cite{Dimakis10} 
which has the form
	\begin{equation}\label{eq:sbt}
	M \le \sum_{i=1}^k\min \{l,(d-i+1)\beta\}.
	\end{equation}
	The repair problem has been studied in two versions, called functional and exact repair. Under exact repair, the contents of the failed node is recovered in the exact form, while for functional repair the node can be restored to a different value as long as it continues to support the two properties above.

	\subsection{Heterogeneous and graph-constrained storage systems} Most earlier works on regenerating codes, with a few exceptions mentioned below, assume that any $d$ of the surviving $n-1$ nodes can
	serve as helpers, and the choice of a particular helper set is not addressed in the code design. Existing code constructions typically also assume
	that downloading the same amount of data from each of the helpers minimizes the communication complexity of 
	repair, and in many cases (e.g., for MSR codes) one can show that this is indeed true. Constructions with 
	nonuniform download are considered when the transmission cost from different helpers to the failed node is 
	not the same, giving rise to {\em heterogeneous regenerating codes}. Different versions of such systems
	include models with two subsets of nodes assigned different communication costs \cite{Akhlaghi2010},
	systems with different link capacities \cite{wang2014heterogeneity}, rack-aware storage \cite{Hou2018,Chen2020,gupta2022rack,chen2022rack,wang2023rack,wangchen2023rack}, general models of clustered systems \cite{prakash2018storage, SohnChoYooMoo2018}, systems with varying amounts of data downloaded from the helper nodes
	\cite{wang2020storage}, \cite{Li2022}, and general systems with topology-aware repair \cite{yang2022hierarchical}.

In \cite{Patra2022,patra2022interior,patra2023node} we considered a related but different problem when the nodes of the 
storage system are placed on the vertices of a graph $G=(V,E)$, and the repair information is downloaded along the edges of the graph.  In this model, the cost of sending a unit of information from $v_i$ to $v_j$ is determined by the graphical distance $\rho(v_i,v_j)$ 	in $G$. This constraint introduces two new features. First, there is a natural bias in the information cost of node repair in favor of the helper nodes closer to the failed node $v_f,$ suggesting that coding schemes should rely on heterogeneous codes. Secondly, since the information from remote helpers passes through the helpers close to $v_f$, the closer nodes have the option of combining the information received from the outer extremes of the helper set before transferring it to the failed node. We call this approach {\em Intermediate Processing}, or IP, as opposed to direct relaying, which we term {\em Accumulate-and-Forward} (AF) strategy. The general IP procedure applies to linear regenerating codes, although using it for  specific code families requires further analysis of their structure. In this paper we closely examine both questions.
	
	Prior to our works, repair on graphs using MSR codes was considered in \cite{GeramiXiao2014,LuXuanFu2014} for particular examples of graphs. These papers relied on AF repair and did not consider incorporating IP repair in the node recovery procedure.
	Another related communication problem is that of network coding \cite{Yeung2006} wherein (in its simplest version) 
	the data is transmitted from a single fixed source to multiple destinations, and where it is assumed that the 
	intermediate nodes combine the chunks of data on their incoming edges. A repair-tree model of node recovery was previously considered in \cite{zhang2016repair} in the context of network coding. While intermediate processing is a shared feature between node repair on graphs and network codes, they address different objectives. Specifically, while in network coding, the message of the source needs to be reproduced at the receiver, the repair task is to compute a function of the cumulative contents of the helpers at the failed node. Nevertheless, some tools from network coding prove useful in the problem of node repair with errors, introduced below in Section \ref{subsec:intro_repair_error}.
		\subsection{Generalized regenerating codes}
	Motivated by the graph model of the system discussed above, we analyze the problems of optimizing the repair
	bandwidth through nonuniform download, varying the repair degree, IP repair, and repair in the presence of an adversary. In graphical networks,
	sending the repair data along a path from the helper to the failed node incurs the cost proportional to the path length, which naturally introduces heterogeneity based on the graphical distance. Accordingly, we extend the definition of regenerating codes
	as follows.
	\begin{definition}[Generalized regenerating codes (GRCs)]\label{def:GRC}
		Let $\cB = \{\beta_i\}_{i=1}^d$ be a set of $d$ positive integers. An $[n,k,d,l,\cB,M]$ GRC encodes a file $\cF$ of size $M$ symbols over a finite field ${F}$ by storing $l$ symbols in each of the $n$ nodes such that
		\begin{enumerate}
			\item (reconstruction) the original file can be recovered by accessing any $k$ out of $n$ nodes;
			\item (repair) the contents of any node $f\in [n]$ can be recovered by contacting a set $D \subseteq [n]\setminus\{f\},|D|=d$ of nodes and downloading $\beta_i$ symbols from node $\tau^{-1}(i)$, for any bijective mapping $\tau: D \rightarrow [d]$.
		\end{enumerate}
		The parameter $d$ is called the {\em repair degree}.    
	\end{definition}
	The mapping $\tau$ corresponds to the allocation of contributions for repair to the set of the helpers, and it highlights the fact that the assignments can be arbitrary as long they form the set $\cB.$ Specifically, the amount of data downloaded from a given helper node may change from one repair round to another depending on its distance to the failed node.  If all the $\beta_i$'s are equal, we call such a repair procedure a {\em uniform download} scheme.	
		
	\subsection{Repair in the presence of adversarial nodes}\label{subsec:intro_repair_error}
The simplest model of errors in a graphical storage system arises from assuming  unreliable links across the 
network. This problem has been widely analyzed in network coding literature; see \cite{Beemer2023} and references therein. 
For regenerating codes, a straightforward way of handling noise edges, if one assumes that a given edge does not introduce 
more than a fixed number of errors, is to simply encode each transmission using a local error 
correcting code with sufficient minimum distance and placing a decoder on every node.
   
 The error process is more difficult to handle if the network includes adversarial nodes that try to interfere with the repair by introducing errors that can affect its outcome. If the information residing on a node is corrupted, 
	it can spread through the network if this faulty node is chosen to be a helper. Previous works on error control 
	during repair (\!\cite{RSRK2012}, \cite{Ye16a}, \cite{Silberstein2015} and others all focus on the traditional model of direct connectivity. These 
	schemes still work in the graph scenario if the nodes rely upon standard relaying of data, i.e., the AF strategy. 
	At the same time, if the nodes perform intermediate processing, error amplification can happen, similar to what 
	happens in network coding with errors. This is because even a single corrupted symbol can potentially affect all 
	the linear combinations evaluated at the node. If left unchecked, data corruption can quickly spread through the network in the course of multiple failures.
	
Our main results for repair with errors are related to the model of a \emph{limited-power adversary}. We assume that the adversary may alter the contents of a subset of nodes $\Delta \subset D$ (not known to the system), but
has no access to the computations performed at any of the nodes. In other words, in this model the helper nodes perform valid transformations of the data, but their own contents may be corrupted without their knowledge. 

It is also possible to consider an \emph{all-powerful adversary} that actively controls a vertex and can alter both the stored data and data that is relayed through that vertex. Coding in the presence of such an adversary appears to be a more difficult problem \cite{KK2008}, \cite{KosutTongTse2014}, and we do not consider it here.

	\subsection{Overview of the results}
In Section \ref{sec:BRB} we derive a bound for the repair bandwidth (the cutset bound), which includes previously derived results \cite{Akhlaghi2010,Patra2022} as special cases. We also present a {\em stacking code construction}, similar to one in \cite{Li2022}, showing that it attains the cutset bound at the MSR point. These codes also support flexible download assignment mentioned above. Note that, while the results of \cite{Patra2022,patra2022interior} show that it is possible to perform repair on graphs with download cost smaller than direct relaying, they do not support the varying download property.

In Sec.~\ref{sec:IPrepair} we turn our attention to IP repair under the GRC framework. We prove a lower bound on the amount of information sent by any subset of helpers for repair, which serves as a benchmark for IP repair in the network. This bound also extends previous results established for the case of uniform download \cite{Patra2022,patra2022interior}. We also show that the stacking construction proposed in Sec.~\ref{sec:BRB} attains the mentioned lower bound with equality. Next, we present a simple linear-algebraic argument showing that any family of linear GRCs supports IP repair.  The final result of Sec.~\ref{sec:IPrepair}  quantifies savings in repair bandwidth that result from using the IP repair scheme.

In Sec.~\ref{sec:Optimizing} we turn to the problem of optimizing the download amounts from 
different helpers based on their distance to the failed node. An application of the result
of \cite{Li2022} shows that, under both IP and AF strategies, if the number of helpers can be dynamically adjusted once we know the identity of
the failed node, there always exists a uniform download scheme that minimizes the overall download. With this insight, we show that achieving the optimal download cost amounts to a combination of choosing the optimal repair degree while also carefully maximizing the benefits of IP. We further show that for codes with high rate, the best AF strategy is to involve all the remaining nodes in the repair process.

In Section \ref{sec:IPC} we present realizations of IP repair under the uniform download assumption for the family of product-matrix codes of \cite{Rashmi11} and their generalization given in \cite{Duursma2020}. 

In Sec.~\ref{sec:errors}, we consider the repair task under the possible presence of adversaries where the main goal is to counter the spreading of errors while performing IP repair on the graph. Extending the Singleton bound from network coding, we derive a bound on the repair bandwidth in the presence of errors for IP repair. We also present a construction that supports repair of systematic nodes in the presence of a limited-power adversary in the network.

	\section{Bounds on the Repair Bandwidth}\label{sec:BRB}
\subsection{The cutset bound}	In this section, we present a general form of the communication complexity bounds for the nonuniform download case. Previously this question was discussed in \cite{Quan2015} based on the information flow graph approach, but the bound we obtain here is easier to interpret and use.
For a finite field ${F}=\ff_q$, we consider a code $\cC\subset {F}^{nl}$ whose codewords $(C_i,i=1,\dots,n)$ are represented by $l\times n$ matrices over ${F}$. We assume that each coordinate (a vector in ${F}^l$) is written on a single storage node, and that a failed node amounts to having its coordinate erased. 
	
	Suppose that the information stored at the nodes is described by random variables $W_i,i\in[n]$ that have some joint distribution on $({F}^l)^n$ and satisfy $H(W_i)=l$ for all $i$, where $H(\cdot)$ is the entropy.  For a subset $A\subset [n]$ we write $W_A=\{W_i,i\in A\}.$
	We assume that $H(\cF|W_A)=0$ for any $A\subset[n],|A|=k,$ which supports the data retrieval property. Let $f \in [n]$ be the failed node and let $D \subseteq [n]\setminus\{f\}, |D|=d$ be the set of helper nodes. Let $S_i^f$ be the information provided to the failed node $f$ by the $i$th helper node in the traditional fully connected repair 
	scheme, and let $S_A^f = \{S_i^f: i \in A\}$ for any $A \subseteq D$.  By definition we have $H(S_i^f) = \beta_{\tau(i)},$  and
	\begin{equation}\label{eq:reg_code_props}
		\begin{split}
			&H(W_K) = M, K \subset [n], |K| = k\\
			&H(S_i^f|W_i) =0 , i \in D;\;\;
			H(W_f|S_D^f) = 0.
		\end{split}
	\end{equation}
	Let 
	$
	\Delta_r(\cB) = \min_{R \subseteq [d],|R|=r} \sum_{i \in R}\beta_i
	$ 
	denote the sum of $r$ smallest elements from $\cB$.
 
 The following statement gives a general bound for information transmission during repair, extending the results in 
	\cite{Shah2012} and \cite{Patra2022}.
	\begin{theorem}\label{theorem:gen_cutset}
		For an $[n,k,d,l,\cB,M]$ GRC,
		\begin{equation}\label{eq:gen_cutset}
			M \le \sum_{i=0}^{k-1}\min\{l,\Delta_{d-i}(\cB)\}.
		\end{equation}
	\end{theorem}
	\noindent{\em Proof.}
	For any $f \in [n]$, any $D \subseteq [n]\setminus\{f\}, |D|=d$ and any set $A \subset D$, we have
	$$
	H(W_f|S_A^f,S_{D\setminus A}^f) = 0,
	$$ 
	which implies
	\begin{align*}
		H(W_f|S_A^f) &= I(W_f;S_{D\setminus A}^f|S_A^f)\le H(S_{D\setminus A}^f|S_A^f)\\
		&\le H(S_{D\setminus A}^f)\le \sum_{i\in D\setminus A}H(S_i^f)\\
		&= \sum_{i \in D\setminus A}\beta_{\tau(i)}.
	\end{align*}
	Since this is true for any bijective mapping $\tau$, we conclude that 
	\begin{equation}\label{eq:H-H}
		H(W_f|W_A) \le H(W_f|S_A^f) \le \min\{l,\Delta_{d-|A|}(\cB)\}.
	\end{equation}
	Finally,
	$$
	M = H(W_{[k]})= \sum_{i=1}^{k}H(W_i|W_{[i-1]})
	\le \sum_{i=0}^{k-1}\min\{l,\Delta_{d-i}(\cB)\}.
	$$

	\begin{remark}
		For the case of uniform download, the bound of this theorem reduces to the inequality $H(W_i|W_A)\le\min(l,(d-|A|)\beta$, where $A\subset[n], |A|\le d, i\not\in A$, proved in \cite{Shah2012}. If the set $\cB$ contains only two distinct values, then \eqref{eq:gen_cutset} recovers the main result of \cite{Akhlaghi2010}, obtained here with a much shorter proof.
	\end{remark}
	As in the case of homogeneous systems \cite{Dimakis10}, we define the following corner points of the above trade-off curve.
\begin{definition}\label{def:corner_points}
		The {\em Minimum Storage} (MSR) point of the bound
		\eqref{eq:gen_cutset}, is defined by $l=\Delta_{d-k+1}(\cB)$, and any code that achieves these parameters for a given set $\cB$ is called an MSR code. Similarly, the Minimum Bandwidth (MBR) point is defined by $l = \Delta_d(\cB) = \sum_{i=1}^d\beta_i$.    
	\end{definition}
	
	\subsection{Code construction}\label{sec:stacking}
	In this section, we describe a construction of an $[n,k,d,l,\cB,M]$ GRC that meets the bound of Theorem~\ref{theorem:gen_cutset} at the MSR point. This construction generalizes the construction of \cite{Li2022},
 and is reminiscent of multilevel concatenated codes of \cite{blokh1974coding}.
  We further show here that for any set $\cB = \{\beta_j\}_{j=1}^d$ the constructed code family achieves the minimal possible per-node storage parameter $l$ and hence is MSR-optimal in the sense of Definition \ref{def:corner_points}. In the next section
  we prove that this construction is also optimal for IP repair.
	
	Given a repair degree $d$ and a set of integers $\cB=\{\beta_j\}_{j=1}^d$, we aim to construct a regenerating code that repairs any failed node $f$ by downloading
	at  most $\beta_j$ symbols from node $\tau^{-1}(j)$ for any subset of helper nodes $D \subseteq [n]\setminus \{f\}$ and any permutation $\tau : D \rightarrow [d]$. 
	By relabeling the nodes, we can always assume that the set $\{\beta_j\}$ is sorted in nondecreasing order. Let $\mu=(\mu_1,\dots,\mu_{d-k+1})$ be the binary vector with $\mu_j = \mathbbm{1}_{(\beta_j > \beta_{j-1})}, j=1,\dots,d-k+1$, where $\beta_0:=0$. 
	
\begin{tcolorbox}[width=\linewidth, sharp corners=all, colback=white!95!black]
	\begin{construction}\label{const:stacked}
		Suppose that $\cB=\{\beta_j\}_{j=1}^d$ and $\mu$ are as defined above, and let $S=\{j: \mu_j\ne0\}$. For every $j \in S$ take an 
		MSR code $\cC_j$ with parameters 
		\begin{equation*}
			\begin{split}
				[n,k,d-j+1,l_j=(d-j-k+2)(\beta_j-\beta_{j-1}),\beta_j-\beta_{j-1}, M_j=kl_j].
			\end{split}
		\end{equation*}
		The $[n,k,d,l,\cB,M]$ GRC code is formed by stacking the codes $\{\cC_j\}_{j \in S}$, where $l=\sum_{j\in S}l_j$ 
		and $M=\sum_{j\in S} M_j$. 
	\end{construction}
 \end{tcolorbox}
	
	The intuition behind the construction is as follows. Upon arranging the $\beta_j$s in nondecreasing order, for every $j$ such that
	$\beta_{j}>\beta_{j-1}$, we add to the stack an MSR code with per node download equal to the gap $\beta_j-\beta_{j-1}$ and repair degree
	$d-j+1.$

	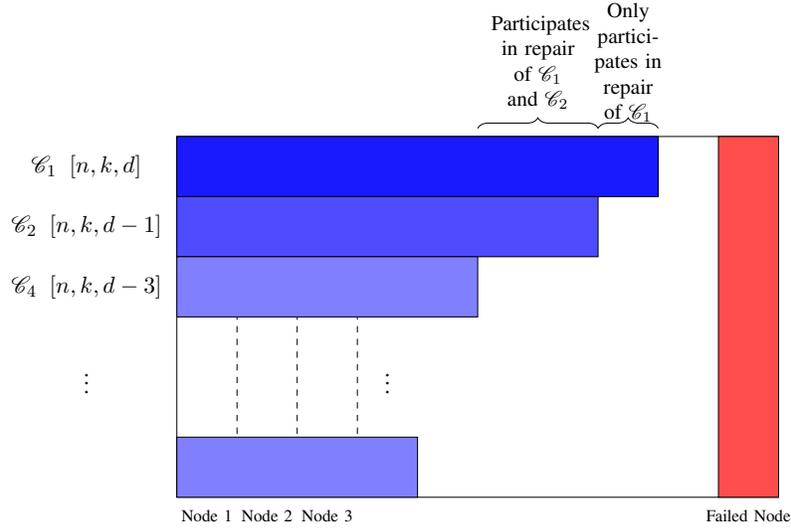
\begin{figure}[th]
		\begin{center}\scalebox{0.8}{
				\begin{tikzpicture}[squarednode/.style={rectangle, draw=white!60, very thick, minimum size=5mm}]
					\draw (0,0) -- (0,6) -- (10,6) -- (10,0) -- (0,0);
					\filldraw[fill=red!70!white, draw=black, fill opacity=0.2] (9,0) rectangle (10,6);
					\draw[dashed] (1,0) -- (1,6);
					\draw[dashed] (2,0) -- (2,6);
					\draw[dashed] (3,0) -- (3,6);
					\filldraw[fill=blue!90!white, draw=black, fill opacity=0.2] (0,6) rectangle (8,5);
					\filldraw[fill=blue!70!white, draw=black, fill opacity=0.2] (0,5) rectangle (7,4);
					\filldraw[fill=blue!50!white, draw=black, fill opacity=0.2] (0,4) rectangle (5,3);
					\filldraw[fill=blue!50!white, draw=black, fill opacity=0.2] (0,0) rectangle (4,1);
					\node[squarednode] at (3.5,2) (dots) {$\vdots$};
					\node[squarednode] at (-1.5,5.5) (C1) {$\cC_1\;\;[n,k,d]$};
					\node[squarednode] at (-1.5,4.5) (C2) {$\cC_2\;\;[n,k,d-1]$};
					\node[squarednode] at (-1.5,3.5) (C3) {$\cC_4\;\;[n,k,d-3]$};
					\node[squarednode] at (-1.5,2) (C3) {$\vdots$};
					\node[squarednode] at (0.5,-0.3) (N1) {\scriptsize Node 1};
					\node[squarednode] at (1.5,-0.3) (N2) {\scriptsize Node 2};
					\node[squarednode] at (2.5,-0.3) (N3) {\scriptsize Node 3};
					\node[squarednode] at (9.5,-0.3) (Nf) {\scriptsize Failed Node};
					\draw [decorate,decoration={brace,amplitude=5pt,raise=4.8ex}]
					(7,5.3) -- (8,5.3) node[midway, yshift=5em,,text width=3.5em, align=center,font=\small\linespread{1.0}\selectfont]{Only participates in repair of $\cC_1$};
					\draw [decorate,decoration={brace,amplitude=5pt,raise=4.8ex}]
					(5,5.3) -- (7,5.3) node[midway, yshift=5em,,text width=4em, align=center,,font=\small\linespread{1.0}\selectfont]{Participates in repair of $\cC_1$ and $\cC_2$};
			\end{tikzpicture}}
			\caption{Example: Stacked MSR construction for $S=\{1,2,4,\cdots\}$.}\centering
			\label{fig_stackedMSR}
		\end{center}
	\end{figure}

	\begin{proposition}\label{lemma:stacked_msr}
		The code in Construction \ref{const:stacked} is an $[n,k,d]$ regenerating code that supports repair of any node $f$ from a helper set $D \subseteq [n]\setminus\{f\}, |D|=d$ by downloading at most $\beta_j$ symbols of ${F}$ from node $\tau^{-1}(j)$ for any permutation $\tau: D \rightarrow [d]$.
	\end{proposition}
	\begin{proof}The length $n$ and the data reconstruction parameter $k$ are the same for all the component codes
 and hence inherited directly. To prove the repair property, fix a helper set $D$ and a permutation $\tau$ such that $\{\beta_j\}$ are in nondecreasing order. 
 Node $\tau^{-1}(j)$ participates in the recovery of node $f$ only in the component codes $\{\cC_p: p \le j\}$ and hence it sends a total of $\sum_{p=1}^j(\beta_p-\beta_{p-1}) = \beta_j$ symbols.
	\end{proof}
	
	The component codes can be chosen from a variety of known MSR constructions. 
	Since the parameters of these outer codes depend on the given set $\{\beta_j\}$, a convenient choice is the product-matrix codes \cite{Rashmi11},
	which in their basic version work by downloading a single symbol from every helper. By stacking several codewords of these codes we can obtain any of the component codes $\cC_j$  as in Fig.~\ref{fig_stackedMSR}, thereby matching any set of per-node download values as required by the construction.
	
	Next, we show that this construction attains the MSR point as per Definition \ref{def:corner_points}.
	\begin{proposition}\label{lemma:gen_msr_sat}
		Codes of Construction \ref{const:stacked} meet the bound \eqref{eq:gen_cutset} with equality at the MSR point.
	\end{proposition}
	\begin{proof}
		The node size for the code of this construction equals
		\begin{align*}
			l &= \sum_{j=1}^{d-k+1}(d-j-k+2)(\beta_j-\beta_{j-1})\\
			&= \sum_{i=1}^{d-k+1}i(\beta_{d-i-k+2}-\beta_{d-i-k+1})=\sum_{i=1}^{d-k+1}\beta_i . 
		\end{align*}
		Since $\beta_j\ge \beta_{j-1}$ for all $j$, summing the first $d$ of the $\beta_j$s gives the minimum value over all permutations $\tau$. Thus,
		the sum on the last line equals $ \Delta_{d-k+1}(\cB),$ matching \eqref{eq:gen_cutset}.
	\end{proof}

	%
	%
	
	\section{IP repair} \label{sec:IPrepair}
In the previous section we considered regenerating codes with nonuniform download in general, without associating them
with repair on graphs. In the first part of this section, we continue this line of thought, deriving a lower bound 
on the minimum required transmission for a set of helper nodes for the repair of a failed node. 
Then we turn to repair on graphs, showing that linear regenerating codes support IP repair on graphs, and that the stacking
code family attains this lower bound.
	
\subsection{Lower bound}
	
In the next theorem, we prove a generalized version of the IP bound,
	extending the results of \cite{Patra2022}. 
	\begin{theorem}\label{thm:gen_IP}
		Let $f \in [n]$ be the failed node. For a subset of helper nodes $E \subset D$, let $R_E^f$ be a function of $S_E^f$ such that $H(W_f|R_E^f,S_{D\setminus E}^f)=0$. Then if $|E| \ge d-k+1$,
		\begin{equation}\label{eq:REf}
			H(R_E^f) \ge M - \sum_{i=0}^{k-2}\min\{l,\Delta_{d-i}(\cB)\}.
		\end{equation}
	\end{theorem}
	\begin{proof}
		Since we assumed that $H(W_f|R_E^f,S_{D\setminus E}^f)=0$, all the more it is true that
		\begin{equation}\label{eq:eqtn2}
			H(W_f|R_E^f, W_{D\backslash E})=0.
		\end{equation}
		We have $|D\backslash E| \le k-1$. Consider a set $A \subset E$ with $|A| = k-1-|D\backslash E|$. Now,
		\begin{equation}\label{eq:eqtn3}
			H(R_E^f, W_{D\backslash E}, W_{A}) = H(R_E^f, W_{D\backslash E}, W_f, W_{A}) \ge M,
		\end{equation}
		where the equality in \eqref{eq:eqtn3} follows from \eqref{eq:eqtn2} and the chain rule, and the inequality follows from the reconstruction property because 
		$|D\backslash E|+|A|+1 = k$. Next, observe that
		\begin{align}\nonumber
			H(R_E^f, W_{ D\backslash E}, W_{A}) &\le H(R_E^f)+H( W_{D\backslash E}, W_{A}),
		\end{align}
		and so 
		\begin{align}\nonumber
			H(R_E^f) &\ge M - H( W_{D\backslash E}, W_{A})\\
			&\ge M- \sum_{i=0}^{k-2}\min\{l,\Delta_{d-i}(\cB)\},\nonumber
		\end{align}
		where the last inequality follows from \eqref{eq:H-H}.
	\end{proof}
	
\begin{corollary}\label{cor:gen_IP_MSR}
(a) Assume that $\beta_i=\beta$ for all $i \in [d]$, then we have
		\begin{equation}\label{eq:uni_IP}
			H(R_E^f) \ge M- \sum_{i=0}^{k-2}\min\{l,(d-i)\beta\}.
		\end{equation}
(b)  For MSR codes, we have
		\begin{equation}\label{eq:H-IP}
			H(R_E^f) \ge l=\Delta_{d-k+1}(\cB)
		\end{equation}
(c) Assuming in addition uniform download with $\beta_i=\beta$ for all $i \in [d]$ at the MSR point, we have
		\begin{equation}\label{eq:H-IP_MSR}
			H(R_E^f) \ge (d-k+1)\beta.
		\end{equation}
	\end{corollary}
	\begin{proof}
Part (a) follows directly from \eqref{eq:REf}. To prove part (b), recall from \eqref{eq:gen_cutset} that
at the MSR point, we have $l= \Delta_{d-k+1}(\cB)$ and $M=kl,$ and use \eqref{eq:REf}. Part (c) is immediate from 
\eqref{eq:H-IP}.
	\end{proof}
Theorem \ref{thm:gen_IP} bounds below the amount of data necessarily obtained from a subset $E\subset D$
	irrespective of the processing performed by the nodes in $E$, including the IP repair on graphs, discussed in the next section. Eq. (\ref{eq:H-IP_MSR}) in Corollary \ref{cor:gen_IP_MSR} is a restatement of Lemma II.1 from \cite{Patra2022}. In that work we also showed that $H(R_E^f)=l$ is achievable at the MSR point for the uniform download case.  Below in this paper we show
	that the value $H(R_E^f)=l$ can be achieved by some code families in the non-uniform download case as well. Note that, in general, for all other non-MSR points of the tradeoff curve, $\Delta_{d-k+1}(\cB)<l.$
	
	The above information-theoretic formulation of the repair problem is also valid for functional repair with a slight modification: for functional repair we only focus on the repair cycle of a single node from $d$ helpers whose repair has already been completed (See Section VI of \cite{Shah2012}). Hence the random variable $W_f$ now denotes the data to be reconstructed at the failed node (which can be different from the data lost prior to this moment.). The next proposition forms an extension of Lemma II.1 in \cite{Patra2022}, generalizing it to all functional and exact regenerating codes at all points of the trade-off curve.
	
\begin{proposition}
			For an optimal functional repair code that meets the cutset bound of Eq.~(\ref{eq:gen_cutset}) with equality, we have $H(R_E^f) \ge \Delta_{d-k+1}(\cB)$. Assuming in addition uniform download with $\beta_i=\beta$ for all $i \in [d]$, we have  $H(R_E^f) \ge (d-k+1)\beta$.
		\end{proposition}
The proof follows straightforwardly by replacing $M$ in \eqref{eq:REf} with the maximum
value of $M$, given by the right-hand side of \eqref{eq:gen_cutset}.

\subsection{Nonuniform download and regenerating codes on graphs} \label{sec:IP-GRC}
 
 Limited connectivity of the network is modeled as placing each node on a vertex of a graph
$G(V,E)$ with $|V|=n,$ where each node has direct access only to its immediate neighbors in $G$. Suppose further that the 
coordinate $C_f$ for some $f\in[n]$ is erased, i.e., that the node $f\in [n]$ has failed. Below we denote the vertex in $V$ that corresponds to $f$ by $v_f$ and use $f$ and $v_f$ interchangeably. 
Suppose further that the set $D\subset[n]$ of helper nodes is fixed and consists of the immediate neighbors of $f$ as well as
of some other nodes. To accomplish the repair, the helper nodes provide information which is communicated to $f$  over the edges in $E_{f,D}$. Each helper node in the graph, starting from the nodes farthest from the failed node, sends data to the next node along the shortest path towards $f$. These nodes can be found by a simple breadth-first search on $G$ starting at $f$. Denote by $G_{f,D} = (V_{f,D},E_{f,D})$ the subgraph spanned by $\{f\}\cup D,$ and let $T_{f,D}$ be a spanning tree 
of this subgraph with $f$ as the root. Let $t= \max_{v \in D} \rho(v,f)$ be the height of this tree. We will use the following notation for spheres and balls around $f$ in $G_{f,D}:$
	$$\Gamma_j(f)=\{v\in V_{f,D}:\rho(v,f)=j\},\; N_i(f)=\cup_{j=1}^i \Gamma_j(f),
	$$ 
and we refer to the vertices in $\Gamma_{j}(f)$  as the helper nodes in \textit{layer} $j$. We have $V(T_{f,D})=\sqcup_{i}\Gamma_i(f)$, where we denote $d_i:=|\Gamma_i(f)|=d_i$, with $\sum_i d_i = d$.
To accomplish the repair, the helper nodes 
provide information which is communicated to $f$  over the edges in $E_{f,D}$. The overall repair bandwidth for the repair of $f$ is 
the total amount of transmitted data across all edges. 

Note that there can be multiple possible choices of $G_{f,D}$ and $T_{f,D}$ and the communication complexity of repair depends on this 
choice. The analysis can be made more concrete if we assume that $G$ satisfies certain regularity assumptions.
	
	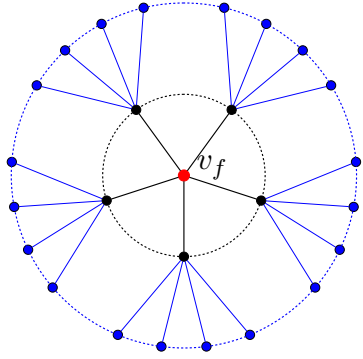
\begin{figure}[th]\begin{center}\scalebox{0.27}{\begin{tikzpicture}
					[
					vertex_style/.style={circle, draw, fill,minimum size=0.08cm,scale=1.25},
					vertex_style1/.style={circle, draw, fill=blue,minimum size=0.08cm,scale=1.25}
					]
					
					\useasboundingbox (-10,-10) rectangle (10,10);
					
					\begin{scope}[rotate=90]
						
						\node[circle,draw=red,fill=red,minimum size=0.1cm,scale=1.5,label={[xshift=1.4cm, yshift=-1cm,minimum size=0.5cm,color=black,scale=4]$v_f$}] (0) at (canvas polar cs: radius=0cm,angle=0){};
						
						\foreach \x/\y in {36/1,108/2,180/3,252/4,324/5}{
							\node[vertex_style] (\y) at (canvas polar cs: radius=4cm,angle=\x){};
						}
						\foreach \x/\y in {0/1,0/2,0/3,0/4,0/5}{
							\path[-] (\x) edge [black,ultra thick] (\y);
						}
						
						\foreach \x/\y in {13.5/6,28.5/7,43.5/8,58.5/9}{
							\node[vertex_style1] (\y) at (canvas polar cs: radius=8.5cm,angle=\x){};
							\path[-] (1) edge [blue,thick] (\y); 
						}
						\foreach \x/\y in {85.5/10,100.5/11,115.5/12,130.5/13}{
							\node[vertex_style1] (\y) at (canvas polar cs: radius=8.5cm,angle=\x){};
							\path[-] (2) edge [blue,thick] (\y); 
						}
						\foreach \x/\y in {85.5+72/14,100.5+72/15,115.5+72/16,130.5+72/17}{
							\node[vertex_style1] (\y) at (canvas polar cs: radius=8.5cm,angle=\x){};
							\path[-] (3) edge [blue,thick] (\y); 
						}
						\foreach \x/\y in {85.5+144/19,100.5+144/20,115.5+144/21,130.5+144/22}{
							\node[vertex_style1] (\y) at (canvas polar cs: radius=8.5cm,angle=\x){};
							\path[-] (4) edge [blue,thick] (\y); 
						}
						\foreach \x/\y in {85.5+216/23,100.5+216/24,115.5+216/25,130.5+216/26}{
							\node[vertex_style1] (\y) at (canvas polar cs: radius=8.5cm,angle=\x){};
							\path[-] (5) edge [blue,thick] (\y); 
						}
						\draw[black,ultra thick,dashed,label=x] (0,0) circle (4.0cm) ;
						\draw[blue,ultra thick,dashed] (0,0) circle (8.5cm);
					\end{scope}
			\end{tikzpicture}}
			\caption{Node repair on a graph: The failed node $v_f$ and the set $D$ of helper nodes, forming the repair tree $T_{f,D}$.}\label{fig:graph1}
		\end{center}
	\end{figure}
 
	\begin{example}\label{example:Cayley} As an example, suppose $G(V,E)$ is a connected $t$-regular graph. One way to guarantee the existence of a convenient repair tree is
		to consider graphs with girth $g$, in which case a ball of radius $\lfloor g/2\rfloor-1$ around any vertex is a tree with $t$ immediate neighbors of the center and $t(t-1)^{i-1}$ vertices in layer $i$. A line of work starting with Margulis's paper \cite{margulis1982explicit} yielded constructions of 
		such graph families with $g\ge C(n,t)\log_{t-1}n$, where $n$ is the number of vertices and $C(n,t)$ is a constant. 
We give concrete examples of repair on these graphs below.
	\end{example}
	
	\subsection{IP repair for linear regenerating codes}\label{sec:general}
In this section we present a general framework for IP repair for any linear GRC, establishing the following result.
The {\em repair tree} of the failed node $f$ is a spanning tree of $G_{f,D}$ with the root at $f$. 
 \begin{theorem}\label{thm:gen_IP_sat} 
(a) {\rm (Existence of IP repair procedure)} Suppose that we are given an $F$-linear regenerating code, a failed node $f$, and a helper node $h$ that is on a path from a subset $B\subset D$ to $f$ in the repair tree $T_{f,D}$. There exists an $F$-linear map that enables 
one to combine the information from the nodes in $B \cup \{h\}$, resulting in an $l$-dimensional vector sent to $f$  to complete the repair. 

(b)	{(\rm Optimality of the stacking construction)} The IP repair procedure for codes of Construction \ref{const:stacked} meets the lower bound \eqref{eq:H-IP}
with equality.
	\end{theorem}

 \begin{proof}
(a) Fix the permutation $\tau: D \rightarrow [d]$.
	Referring to the notation of the previous section, our task is to recover the contents of $W_f$ by downloading $\beta_{\tau(h)}$ symbols from the helper node $h$. Denote by $S_{h,f}^{\tau} \in F^{\beta_{\tau(h)}}$ the symbols of $F$ provided by $h$ for repair, and suppose
	that $S_{h,f}^{\tau}=\cG_{h,f}^{\tau}(W_h),$ where $\cG_{h,f}^{\tau}: F^l \rightarrow F^{\beta_{\tau(h)}}$ is an $F$-linear map determined by the code. By the definition of the repair property, we know that if the symbols $S_{h,f}^{\tau}$ are available to $f$ (say, if $G$ is a complete graph) then there exists an $F$-linear map $\sF^{\tau}_{f,D}: F^{\sum_{j=1}^d\beta_j} \rightarrow F^l$ such that
	$$ 
	W_f=\sF^{\tau}_{f,D}(\cG^{\tau}_{h_1f}(W_{h_1}),\cG^{\tau}_{h_2f}(W_{h_2}),\cdots,\cG^{\tau}_{h_df}(W_{h_d})).
	$$ 
	For notational simplicity, from now on we suppress the dependence of the maps and the random variables on the permutation
 $\tau$. The above equation implies that there exists a matrix $U_{f,D} \in F^{l \times \sum_{j=1}^d\beta_j}$ such that 
	\begin{equation}\label{eq:WU}
		W_f = U_{f,D} \cdot [
		(S_{h_1,f})^\intercal \mid (S_{h_2,f})^\intercal\mid\dots   \mid(S_{h_d,f})^\intercal)
		]^\intercal.
	\end{equation}
	Writing the matrix $U_{f,D}$ in block form as
	$[ U_{h_1,f} \, U_{h_2,f} \, \cdots\,  U_{h_d,f}]$, we can rewrite \eqref{eq:WU} as
	\begin{equation}\label{eq:3}
		W_f = \sum_{h \in D}U_{h,f}S_{h,f}
	\end{equation}
	where $U_{h,f} \in F^{l \times \beta_{\tau(h)}}$ and $S_{h,f} \in F^{\beta_{\tau(h)}}$.
	Thus, there exist linear maps $I_{h,f}: F^{\beta_{\tau(h)}} \stackrel{U_{h,f}}\longrightarrow F^{l}, h\in D.$
	Since the matrices $U_{h,f}$ do not depend on the codeword, the value $U_{h,f}S_{h,f}$ can be computed at any part of the network 
	by any node with access to $S_{h,f}$. This shows that any set of helper nodes $A \subseteq D$, instead of sending
	$\{S_{h,f}: h \in A\},$ which requires $\sum_{h \in A}\beta_{\tau(h)}$ transmissions, can pass along $\sum_{h \in A}U_{h,f}S_{h,f},$ which requires $l$ transmissions. 
	Starting from the leaf nodes in $G_{f,D}$, the data is moved toward $f$, and as the set $A$ of the already involved helpers increases in size, switching to the latter mode involves savings in the repair bandwidth.

(b)	Observe that codes of Construction \ref{const:stacked} are formed of $F$-linear MSR codes, so they are themselves $F$-linear, and therefore support IP repair. Clearly, they also minimize the amount of data sent by any subset of $d-k+1$ nodes at the MSR point.
	
		For a given $j$ and a code $\cC_j$ it is possible to perform IP repair. Specifically, any subset of at least $d-j-k+2$ nodes can perform intermediate processing for $\cC_j$ to compress their repair data to $l_j$ symbols of $F$. Therefore overall, the subset of nodes of size $d-k+1$ or more
		can perform IP repair, compressing their data to $\sum_{j\in S} l_j=l$ symbols.
	\end{proof}

In Sec.~\ref{sec:IPC} we give an example of implementing this procedure for two related families of regenerating codes. More examples of this kind are given in \cite{patra2022interior,Patra2022NodeRO}.

\subsection{Repair bandwidth gains with nonuniform download} 
	
The repair procedure for MSR codes that optimizes the overall communication complexity of repair with uniform download was 
presented in \cite{Patra2022}. It involves transporting the helper data towards the failed node, i.e., the root of the tree, along the 
edges of the tree, whereby nodes having more than $d-k+1$ children process the information and send $l$ symbols relying on the IP 
technique. Let $J$ be the set of nodes in $T_{f,D}$ with at least $d-k+1$ children and let $J_i=\Gamma_i(f)\cap J$ be the set of nodes 
in $J$ at distance $i$ from the root. For an $i \notin J$, let $\cP(i)$ denote the nearest parent of node $i$ in $J$, and if no such 
parent exists, then let $\cP(i) = f$. 	

Clearly, the set $J$ does not change when we switch from the uniform download model to the nonuniform one and vice versa. Furthermore, every node in $J$ keeps transmitting $l$ symbols by relying on the IP procedure.
Indeed, if a node has $d-k+1$ or more children in the tree, they jointly must transmit at least $l$ symbols for repair because of the bound \eqref{eq:H-IP}, irrespective of whether the $\beta_i$'s are equal or different. 

Assume now that nodes in layer $i$ each contribute $\beta_i$ symbols for repair with $\beta_1 \ge \beta_2 \ge \cdots \ge \beta_t$. This can be accomplished by using an $[n,k,d,l,\cB,M]$ MSR code from Construction \ref{const:stacked} with the set $\cB$ formed of $\beta_i$'s, each appearing
$d_i$ times, for all $i \in [t]$. Let $\delta_i = \beta-\beta_i$ where $\beta = \frac{l}{d-k+1}$ is the uniform download value at the MSR point for the same per-node storage $l$. Furthermore, let 
   $$
   t':=\max\Big(s\in[t]: \sum_{i=s}^{t}d_i \ge d-k+1\Big).
   $$
It can be checked that the vector $\mu$ of Construction \ref{const:stacked} in this case is given by
$$
\mu = e_1+\sum_{i=  t'}^{t} e_{x_i},
$$ 
where $x_i=\sum_{j=i}^t d_j+1$ and $e_x\in\{0,1\}^{d-k+1}$ is a vector with a single 1 in position $x$.

We denote by $\Lambda(T_{f,D})$ the total communication complexity of repairing node $f$ using the repair tree $T_{f,D}$. Additionally, we use superscripts IP and AF for repair with the IP and AF schemes, and 
subscripts U and NU for uniform and nonuniform download. For instance, $\Lambda_{\text U}^{\text IP}$ refers to
the communication complexity of repair with IP and uniform download.

The next theorem states conditions 
under which the nonuniform download model attains bandwidth gains over the uniform one.
	\begin{theorem}\label{theorem:benefit_nonuni}
For a fixed $\cB = \{\beta_i\}$ and $\tau: D \rightarrow [d]$, suppose that a codeword of an $[n,k,d,l,\cB,M]$ GRC 
at the MSR point is encoded on the vertices of a graph $G$. There exists an explicit  scheme that has the following communication complexity of repair for the node $f$ with repair tree $T_{f,D}$:
		\begin{equation}\label{eq:bw_nonuni}
			\Lambda_{\text{\rm NU}}^{\text{\rm IP}}(T_{f,D}) = \sum_{i\in J\setminus\{f\}}l +\sum_{i=1}^t\sum_{j \in \Gamma_i\setminus J_i}\rho(j,\cP(j))\beta_i.
		\end{equation}
	Furthermore, the nonuniform contribution model attains savings over the uniform one whenever
	\begin{equation}\label{eq:P}
		\sum_{i=1}^t\sum_{j \in \Gamma_i\setminus J_i}\rho(j,\cP(j))\delta_i > 0
	\end{equation}
	subject to
	\begin{equation}\label{eq:delta0}
		\sum_{i=t}^{  t'+1}d_i\delta_i+\Big(d-k+1-\sum_{i=t}^{  t'+1}d_i\Big)\delta_{  t'}=0.
	\end{equation}
	\end{theorem}
 
	\begin{proof}
By Theorem~\ref{thm:gen_IP_sat}(a), every node in $J$ may transmit only $l$ symbols. Expression 
\eqref{eq:bw_nonuni} is obtained simply by accounting for the number of symbols transmitted by
each node. The condition in Eq.~(\ref{eq:P}) is obtained by comparing the expressions for uniform and nonuniform download. To obtain \eqref{eq:delta0}, recall
		our notation $\Delta_{d-k+1}$ defined before Theorem~\ref{theorem:gen_cutset}. For the graph case considered, it has the following form:
		$$
		\Delta_{d-k+1}=\sum_{i=  t'+1}^{t}d_i\beta_i+\Big(d-k+1-\sum_{i=  t'+1}^td_i\Big)\beta_{  t'}=l,
		$$
		where the last equality follows from Theorem~\ref{theorem:gen_cutset}.   Rewriting this using the $\delta_i$'s, we obtain   
		Eq.~\eqref{eq:delta0}. 
	\end{proof}
	
\begin{remark}Note that the set $J$ may not include all the nodes capable of performing IP. Indeed, for a choice of 
$\cB = \{\beta_i\}$, any node in the repair tree that accumulates the repair data of a set $A$ such that 
$\sum_{i \in A} \beta_i \ge l$ can gainfully perform IP. Hence, the minimum communication complexity of repair can potentially be even lower than Eq. (\ref{eq:bw_nonuni}).
\end{remark}

Ways of applying Theorem \ref{theorem:benefit_nonuni} depend on the structure of the specific graph family. One such example is given next.
  \begin{example}\label{ex:nonuni}
Consider the $t$-regular Cayley graphs mentioned in Example \ref{example:Cayley}. Suppose that the repair 
tree $T_{f,D}$ is formed of $a$ layers, where $a<g/2,$ then 
     $$
d_i=t(t-1)^{i-1}, i\le a-1 \text{ and } d_a=d-\sum_{i=1}^{a-1} t(t-1)^{i-1}.
    $$ 
    Suppose further that $d_a+d_{a-1}\ge d-k+1.$ 
To simplify the analysis, we are not
including IP since it is somewhat independent of the current discussion and can be easily incorporated into it. The 
overall repair bandwidth for a uniform contribution repair scheme for an $[n,k,d,l,\beta=\frac{l}{d-k+1},M]$ MSR code is  $\Lambda_{\text{U}}^{\text{AF}}(T_{f,D}) = \beta\sum_{i=1}^{a}i d_i$.
			Now let us switch to the nonuniform scheme with helper nodes in layer $i$ contributing $\beta_i$ symbols each, with $\beta_i$'s nonincreasing. From \eqref{eq:delta0}, we have that $d_a\delta_a + (d-k+1-d_a)\delta_{a-1}=0$, with
			$\delta_i=\beta-\beta_i$, and the repair bandwidth under this scheme is
			$
			\Lambda_{\text{NU}}^{\text{AF}}(T_{f,D}) = \sum_{i=1}^{a}i d_i \beta_i.
			$ 
			Note that if $\delta_a >0$ then $\delta_{a-1} <0$ and $\delta_i\le \delta_{a-1}$ 
			for all $i\le a-2,$ so we let $\delta_i= -\frac{d_a}{d-k+1-d_a}\delta_a$
			for all $i \le a-1$ and observe that the savings in the nonuniform setting are
			\begin{align} 
				\Lambda_{\text{U}}^{\text{AF}}(T_{f,D}) - \Lambda_{\text{NU}}^{\text{AF}}(T_{f,D}) &=  \sum_{i=1}^a id_i\delta_i \notag\\ 
				&= \frac{d_a\delta_a}{d-k+1-d_a}\Big(\sum_{i=1}^{a-1}(a-i)d_i-a(k-1)\Big). \label{eq:example_eq}
			\end{align} 
In summary, using the nonuniform scheme results in savings whenever the expression in the parentheses is positive, which is possible for small $k$.
	\end{example}
We end this section with a remark on the generalization of some of the results.   Using other interior point exact-repair code families as component codes in the stacking construction of Section \ref{sec:stacking}, one can get non-MSR GRCs. Existence of IP repair for such codes still follows from Part (a) of Theorem \ref{thm:gen_IP_sat}, however, optimality, i.e., Part (b) of Theorem \ref{thm:gen_IP_sat} might not hold anymore. Conditions for the advantage of the nonuniform download scheme over the uniform one for such interior point GRCs can be stated similarly to Theorem \ref{theorem:benefit_nonuni}. 

	\section{Optimizing the helper data and the repair degree}\label{sec:Optimizing}
 In the analysis up until now, we have focused on the problem of repair on graphs with a fixed repair degree $d$. We showed that the framework of GRCs along with the IP technique give nontrivial ways to minimize the overall communication complexity of node repair in the graph constrained setting by lowering the contribution of the farthest away nodes at the expense of increasing the contributions from the nearer nodes. This gives rise to the question of the limits of this exchange. In the limiting case, one might stop accessing data from the farthest nodes altogether, effectively \emph{decreasing} the repair degree $d$ of the repair process. Indeed, the repair degree need not be a fixed parameter and may be dynamically adjusted if necessary. Universal constructions of regenerating codes proposed in \cite{Ye16a}, \cite{liu2022optimal} support the option of such dynamical adjustment. 

 At the same time, the original work on regenerating codes \cite{Dimakis10} showed that (in the fully connected setting), the repair bandwidth is a decreasing function of the repair degree. This implies that \emph{increasing} the repair degree reduces the communication complexity. Based on these contrasting observations, a natural question to ask is what is the optimal choice of the repair degree under the communication constraints described by a graph.

In this section we find the minimum repair bandwidth obtained by optimizing the repair degree and the download amounts with
GRCs introduced in Construction~\ref{const:stacked} above. We also show that under the AF repair scheme (i.e., with no intermediate processing), for certain parameter regimes, the overall repair bandwidth decreases as more and more helper nodes are involved and hence the optimal choice of $d$ is $n-1$. This result is established for both deterministic and random graphs.
 
	\begin{example}\label{ex:nonuni_cont} To motivate the discussion, let us return to Example \ref{ex:nonuni}, showing that  
 the idea of adjusting the repair degree naturally arises from GRCs. We will show that the maximum savings can be 
		attained by adjusting the repair degree and switching to the uniform assignment.  Expression \eqref{eq:example_eq} implies that, as we increase $ \delta_a= \beta-\beta_a$ above, the advantage of the nonuniform assignment $\Lambda_{\text{U}}^{\text{AF}}(T_{f,D}) - \Lambda_{\text{NU}}^{\text{AF}}(T_{f,D})$ increases, attaining the maximum when $\delta_a= \beta$
or $\beta_a =0$. At this point, 
		\begin{align*}
			\beta_i&= \beta - \delta_{a-1} = \beta + \frac{d_a\delta_a}{d-k+1-d_a}\\
			&= \frac{l}{d-d_a-k+1} = \frac{l}{d'-k+1}, {1\le i\le a-1,}
		\end{align*}
where $d' = d-d_a$ is the new repair degree upon discarding the nodes in layer $a$. 
Note also that every helper node in layers $a-1$ and below contributes 	equally. In summary, the maximum savings are obtained with 
a uniform assignment, but with a smaller repair degree.
	\end{example}
	
In the remainder of this section we address the question of the optimal choice of the repair degree using the already established 
framework of GRCs with IP. First, note that for a fixed repair tree, finding the minimum bandwidth can be formulated as a linear 
programming (LP) problem. It was shown in \cite{Li2022} that such a problem for GRCs always has an optimal solution that supports uniform contribution from a possibly smaller set of helper nodes. Leveraging this result, we incorporate the IP repair technique into the optimization, which becomes feasible because we need to consider only the uniform contribution schemes. 

 To form the optimization problem, without loss of generality, we start with $D =[n-1]$ and assume that node $i \in D$ contributes $\beta_i\ge 0$ symbols for the repair of $f$, where $\beta_i=0$ accounts for the variation of the repair degree. 
 By Theorem \ref{theorem:gen_cutset}, we have that $l = \Delta_{n-k}(\cB)$, which imposes constraints on our choice of $\beta_i$'s. 
 Recalling Theorem~\ref{theorem:benefit_nonuni}, given a specific choice of the repair tree $T_{f,D}$, our goal is to minimize 
    \begin{equation}\label{eq:minimize_ip_bw}
		\Lambda_{\text{NU}}^{\text{IP}}(T_{f,D}) = \sum_{i\in J\setminus\{f\}}l +\sum_{j=1}^t\sum_{i \in \Gamma_j\setminus J_j}\rho(i,\cP(i))\beta_i
	\end{equation}
over the choices of $\cB =\{\beta_i\}_{i \in D}$ such that $\Delta_{n-k}(\cB) \ge l$ and $0 \le 
\beta_i \le l$ for all $i \in D$. Note that letting some $\beta_i$'s to be 0 does not change the 
set $J$, since due to the constraint $l\le\Delta_{n-k}(\cB)$, each node in $J$ can still perform 
IP\footnote{This observation shows that the savings from IP do not depend on the chosen value of the repair degree.}. Since the first term on the right in \eqref{eq:minimize_ip_bw} depends only on the set $J$, 
the minimization is restricted only to the second term. Setting the weights in the linear program to $b_i=\rho(i,\cP(i))$, we obtain the following linear program:
	\begin{equation}\label{eq:opt_prob}
		\begin{split}
			\min_{\{\beta_i\}} \displaystyle&\sum_{i \in D\setminus J}b_i\beta_i \\
			\text{subject to} \displaystyle\sum\limits_{i \in A} &\beta_i \geq l, \;\;\forall A \subseteq D, |A| =n-k\\
			&0 \le \beta_i \le l, \;\;i \in D.
		\end{split}
	\end{equation} 
 Below we assume that the costs $b_i$'s are in nonincreasing order relative to $i$, which is always possible by 
 relabeling the nodes in $D$. This LP problem has been studied in \cite{Li2022} where the authors claimed that the optimal solution takes the form given in the next theorem. The proof does not seem to appear in the published literature, so 	we have included it in the Appendix. 
	\begin{theorem}{\rm (\!\cite{Li2022}, Theorem~1])}\label{theorem:opt_sol}
		There exists an optimal solution of the above LP such that
		\begin{equation}\label{eq:opt_sol}
			\beta^*_i = \begin{cases}
				0 & 1\le i \le n-d-1\\
				\frac{l}{d-k+1} &n-d\le i \le n-1
			\end{cases}
		\end{equation}
		for some $d$ in the range $k \le d \le n-1$.
	\end{theorem}
This implies that among the uniform download schemes, there exists an assignment of $\beta_i$'s that minimizes the repair bandwidth. Let $\cT_f$ denote the set of all spanning trees of $G$ rooted at $f$. Any rooted spanning tree $T_{f,D}$ with $D \subset [n]\setminus\{f\}$ is a subtree of some element $T_f \in \cT_f$. Let $\sigma(v)$ be the number of descendants of $v\in D$ in the tree $T_{f,D}$. Then the minimum repair bandwidth is given as follows:
\begin{corollary}\label{cor-min} Using the stacking construction of codes and the repair (transmission) scheme found in Theorem~\ref{theorem:benefit_nonuni}, the minimum total communication complexity of repairing the failed node $f$ is
    $$ 
    \min_{T_f \in \cT_f}\min_{\stackrel{T_{f,D} \subseteq T_f}{D: k \le |D| \le n-1}}\Lambda^{\text{IP}}_{\text{U}}(T_{f,D})
    $$
where 
    $$
\Lambda^{\text{IP}}_{\text{U}}(T_{f,D}) = \sum_{h \in D}\min\{\sigma(h)+1, (d-k+1)\}\frac{l}{(d-k+1)}.
   $$
\end{corollary}
This implies that for every failed node, there exists at least one \emph{optimal} repair tree and a corresponding \emph{optimal} set of helpers such that the uniform contribution from them, combined with IP, gives the minimum complexity of repair across all GRC-IP schemes. Furthermore, this optimal choice of helpers can be found in time polynomial in the number of vertices.

	\subsection{Optimizing the repair degree}
 Sometimes using IP repair may be too complicated for the storage system. In this case the nodes rely only on the AF strategy and do not perform intermediate processing. This assumption enables us to further
 simplify the minimization in Cor.~\ref{cor-min}, as shown below. In the following analysis, we will use the notation 
 $$\Gamma_j^G(f)=\{v\in V:\rho(v,f)=j\},\; N_i^G(f)=\cup_{j=1}^i \Gamma_j^G(f),
	$$ 
to denote spheres and balls around $f$ in the graph $G$, with the superscript denoting the difference with the notation used in Section \ref{sec:IP-GRC}.
 
\vspace*{.1in} \subsubsection{\sc Deterministic graphs}
With the above assumption, the repair bandwidth does not depend on the choice of the repair tree, and the weights in the linear
program are simply the distances to the failed node. We obtain the following LP problem:
	\begin{equation}\label{eq:opt_prob_AF}
		\begin{split}
			\min_{\beta_1,\dots,\beta_{n-1}} \displaystyle&\sum_{i\ne f}\rho(i,f)\beta_i \\
			\text{subject to} \displaystyle\sum\limits_{i \in A} &\beta_i \geq l, \;\;\forall A \subseteq\{1,\dots,n-1\}, |A| = n-k\\
			&0 \le \beta_i \le l, \;\;i \in \{1,\dots,n-1\}.
		\end{split}
	\end{equation} 
	By Theorem \ref{theorem:opt_sol}, for some repair degree $d^*_{\text{AF}}$ there exists an optimal solution with uniform node contributions. For a given $d \in \{k,k+1,\dots,n-1\}$, define 
	  $$
   \Lambda_{\text{U}}^{\text{AF}}(d) = \sum_{i=1}^tid_i\frac{l}{d-k+1},
   $$
	where $t$ is such that $|N_t^G(f)|\ge d > |N_{t-1}^G(f)|$, and $d_i = |\Gamma_i^G(f)|$ for $1 \le i \le t-1$, $d_t = d-|N_{t-1}^G(f)|$. Then, the optimal repair degree is given by
	$$d^*_{\text{AF}} = \argmin_{d \in \{k,k+1,\dots,n-1\}}\Lambda^{\text{AF}}_{\text{U}}(d).$$
	
As mentioned before, for complete graphs, the quantity $\Lambda_{\text{U}}^{\text{AF}}(d)$ is a decreasing function of the repair degree and so $d^*_{\text{AF}} = n-1$. We will show that this is also true in general for arbitrary graphs provided that the code rate is large enough. As the first step, we prove that 
involving more than $k$ nearest nodes in the repair process entails saving of the repair bandwidth.

\begin{proposition}\label{prop:U_AF_repair}
		For a repair graph, if $k\in \Gamma_a^G(f),$ then
		$$
  \Lambda^{\text{AF}}_{\text{U}}(k) \ge \Lambda^{\text{AF}}_{\text{U}}(|N_a^G(f)|).
        $$
	\end{proposition}
	\begin{proof}
If $k=|N_a^G(f)|$, then the claim is trivially true, so assume that $|N_a^G(f)|\ge k+1$. 
Let $p=k-|N_{a-1}^G(f)|$. For $d \in \{k,k+1,\dots,|N_a^G(f)|\}$, we have
		$$
  \Lambda^{\text{AF}}_{\text{U}}(d) = \sum_{i=1}^aid_i\frac{l}{d-k+1}.
  $$ 
For any $m \in \{1,2,\dots,|N_a^G(f)|-k\}$, 
		\begin{align*}
			\Lambda^{\text{AF}}_{\text{U}}(k)-\Lambda^{\text{AF}}_{\text{U}}(k+m) = \sum_{i=1}^{a-1}id_il\Big(1-\frac{1}{1+m}\Big)+al\Big(p-\frac{p+m}{1+m}\Big)>0
		\end{align*}
		for $p,m \ge 1$.
	\end{proof}
This claim can be further specified if the underlying graph is regular. We will show that for codes of sufficiently
high rate, the repair bandwidth decreases with the increase of the repair degree. 
	\begin{theorem}\label{lemma:threshold_k}
		Consider repair on a $t$-regular graph. Let $m = \max_{h \in V\setminus\{f\}}\rho(f,h)$ be the maximum height of the repair tree. If 
     $$
     k > 1+\max_{1\le a\le m} \sum_{i=1}^{a-1}t(t-1)^{i-1}(1-i/a),
     $$
then $d^*_{\text{AF}} = n-1$. 
	\end{theorem}
	\begin{proof}
Proposition~\ref{prop:U_AF_repair} implies that $\Lambda^{\text{AF}}_{\text{U}}(d)$ becomes smaller as $d$ is increased from 
$k$ in the corresponding layer. We first show that if $d= \sum_{i=1}^{a-1}t(t-1)^{i-1}$ for some $a$, then $\Lambda^{\text{AF}}_{\text{U}}(d) \ge \Lambda^{\text{AF}}_{\text{U}}(d+1)$, that is involving one extra node from the $(a+1)$-th layer does not increase the overall bandwidth. 
Let
   $$
K(a) = \sum_{i=1}^{a-1}it(t-1)^{i-1}, \quad C(a)= \sum_{i=1}^{a-1}t(t-1)^{i-1}+1.
  $$
Indeed, 
		\begin{align*}
			\Lambda^{\text{AF}}_{\text{U}}(d) - \Lambda^{\text{AF}}_{\text{U}}(d+1) &= \sum_{i=1}^{a-1}it(t-1)^{i-1}\frac{l}{d-k+1} - \Big(\sum_{i=1}^{a-1}it(t-1)^{i-1}+a\Big)\frac{l}{d-k+2} \\
			&=\frac{K(a)l}{C(a)-k} - \frac{(K(a)+a)l}{C(a)-k+1}\\
			&= \frac{K(a)l}{(C(a)-k)(C(a)-k+1)}-\frac{al}{C(a)-k+1}
		\end{align*}
		which is nonnegative whenever $(C(a)-k)a \le K(a)$ or $k \ge C(a)-\frac{K(a)}{a}$.
Now we show that, as we keep involving more nodes in this $a$-th layer, the overall bandwidth can only further decrease. For some $x \in \{1,\dots,t(t-1)^{a-1}\}$, let $d = \sum_{i=1}^{a-1}t(t-1)^{i-1}+x$. Then
		\begin{align*}
			\Lambda^{\text{AF}}_{\text{U}}(d)
			&= \Big(\sum_{i=1}^{a-1}it(t-1)^{i-1}+a x\Big)\frac{l}{d-k+1}\\
			&=\frac{l(K(a)+ax)}{C(a)-k+x}. 
		\end{align*}
The function $f(x) = \frac{K(a)+ax}{C(a)-k+x}$ is a continuous function of $x$ in the interval $[1,t(t-1)^{a-1}]$ and is decreasing on $x$ whenever $(C(a)-k)a < K(a)$ or equivalently $k > C(a)-\frac{K(a)}{a}$. Hence the 
values $\{\Lambda^{\text{AF}}_{\text{U}}(d)\}$ at the points $d \in \{C(a),C(a)+1,\dots,C(a+1)-1\}$ form a decreasing sequence. The result now follows since this analysis applies to all $a \in \{1,\dots,m\}$. 
	\end{proof}
\begin{example} For the sake of example, suppose that the network is given by the Petersen graph, a cubic graph
on $n=10$ vertices with 15 edges. Suppose that a code of length 10 is used to support node repair. The diameter of the graph
is 2, so the height of the complete repair tree is $m=2$. By Theorem \ref{lemma:threshold_k}, if the parameter $k$ of the code is $3$ or more, then $d^*_{\text{AF}} = 9$. 
\end{example}
\begin{example}
Let $i$ be an integer satisfying $i^2 \equiv -1 \pmod{29}$, e.g., $i=12$.  By the LPS construction \cite{Lubotzky1988}, the Cayley graph of the group $\text{PSL}(2,29)$ with generating set
	$$
 \begin{bmatrix}
1 & 2\\
-2 & 1
\end{bmatrix},\begin{bmatrix}
1 & -2\\
2 & 1
\end{bmatrix},
\begin{bmatrix}
		1+2i & 0\\
		0 & 1-2i
	\end{bmatrix},\begin{bmatrix}
	1-2i & 0\\
	0 & 1+2i
\end{bmatrix},
\begin{bmatrix}
1 & 2i\\
-2i & 1
\end{bmatrix},\begin{bmatrix}
1 & -2i\\
2i & 1
\end{bmatrix}$$ 
is a 6-regular graph with $n=12180$ vertices and girth 10. Considering repair on this graph, we see that by Theorem \ref{lemma:threshold_k}, if $k \ge 5000$ then $d^*_{\text{AF}} = n-1$.
\end{example}

Theorem \ref{lemma:threshold_k} can be generalized for arbitrary graphs as follows. 	

	\begin{theorem}
For any node $f$ in the graph $G=(V,E)$, let $m_f = \max_{h \in V\setminus\{f\}}\rho(f,h)$. 
For $a \in\{1,\dots,m_f\}$, let $K(a) = \sum_{i=1}^{a-1}i|\Gamma_i^G(f)|,C(a)= |N_{a-1}^G(f)|+1$. If $k > 
\max_{1\le a\le m_f}(C(a)-\frac{K(a)}{a})$, then for the repair of node $f$, $d^*_{\text{AF}} = n-
1$.
	\end{theorem}
The proof is very similar to the proof of Theorem \ref{lemma:threshold_k}, and is therefore omitted.
 
\vspace*{.1in}	\subsubsection{\sc Random graphs}
	
	In this section, we consider the case when the underlying graph $G(V,E)$ is sampled uniformly from the 
 Erd{\"o}s-R{\'e}nyi ensemble of graphs $\cG_{n,p}$ with $p \in (0,1)$. Denote a random element from 
 the ensemble by $\bG_{n,p}$ and the spheres and balls around a node $f$ accordingly by $\Gamma_i^{\bG}(f)$ and $N_i^{\bG}(f)$. For a fixed value of $d$, the repair problem was considered and analyzed in 
 \cite{Patra2022}, where we established the parameter regimes for which IP repair is beneficial to the 
 AF repair. Following that work, here we assume that $p \gg \frac{\log n}{n}$, which ensures that $\bG_{n,p}$ is connected 
 w.h.p., and that $k = \Theta(n)$ for the code rate to be asymptotically positive. We say that {\em $t$-
 layer repair} of the failed node $f$ {\em is possible} if 
	$$
	\PP(|N_t^{\bG}(f)|\ge d)\to 1 \text{ as $n\to\infty$.}
	$$
	and call the minimum $t$ for which this holds the \emph{threshold depth} for repair.
	We have
	$$\Lambda_{\text{U}}^{\text{AF}}(d) = \frac{l(td- |N_{t-1}^{\bG}(f)|) }{d-k+1}.$$ 
	We will use the following two results regarding the random Erd{\"o}s-R{\'e}nyi graphs (below $\PP=\PP_{\cG_{n,p}}$).
	\begin{lemma}[\!\!\cite{Bollobas81}, p.~50; \cite{FK2016}, Sec.7.1]\label{lemma:diam_random}
		
		(i) If $p^2n-2\log n \rightarrow \infty,$ and $	n^2(1-p) \rightarrow \infty,$
		then $$\PP(\diam(\bG_{n,p})=2)\to 1.$$
		
		(ii) Suppose that the functions $t = t(n) \ge 3$ and $0<p=p(n) <1$ satisfy 
		\begin{gather*}
			(\log n)/t - 3\log \log n \rightarrow \infty, \quad p^t n^{t-1} - 2\log n \rightarrow \infty,\\
			p^{t-1}n^{t-2} - 2\log n \rightarrow -\infty,
		\end{gather*}
		then  $\PP(\diam(\bG_{n,p})=t)\to 1.$
	\end{lemma}
	\begin{lemma}[\!\!\cite{Chung01}, Lemma 3]\label{lemma:Ni_random}
		Suppose that $p \ge \frac{\log n}{n}$. For any $\epsilon >0$ and all $i=1,\dots,\lfloor\log n\rfloor$
		\begin{gather}\nonumber
			\PP(|\Gamma_i^{\bG}(x)| \le (1+\epsilon)(np)^i)\ge 1-{1}/{\log^2 n} \\
			\PP(|N_i^{\bG}(x)| \le (1+2\epsilon)(np)^i)\ge 1-{1}/{\log^2 n}\nonumber . 
		\end{gather}	
	\end{lemma}
We now state the result for the optimal repair degree.	
	\begin{theorem}
	If $(np)^{t-1} = o(n), \frac{(np)^t}{n} - 2\log n \rightarrow \infty$ and $k= \Theta(n)$, then 
	$$\PP(d^*_{\text{AF}}=n-1) \rightarrow 1.$$
	\end{theorem}
	\begin{proof}
	 Define the two events $E = \{|N_{t-1}^{\bG}(f)| < k\}$ and $F=\{|N_t^{\bG}(f)| = n-1\}$.  Since, $(np)^{t-1} = o(n)$ abd $k=\Theta(n)$, by Lemma~\ref{lemma:Ni_random}, 
	 $$
	 \PP(E) 
	 \ge \PP(|N_{t-1}^{\bG}(f)| \le (1+2\epsilon)(np)^{t-1}) \rightarrow 1
	 $$ 
	 and by Lemma~\ref{lemma:diam_random},
	 $$
	 \PP(F) = \PP(|N_t^{\bG}(f)| = n-1) = \PP(\diam(\bG_{n,p})=t)\rightarrow 1
	 $$ 
	 	and so $\PP(E \cap F) \ge \PP(E)+\PP(F)-1 \rightarrow 1$.  For any element under the event $E\cap F$, by definition the $k$-th nearest node from the root $f$ is at distance $t$ from it, and the
   remaining $n-k-1$ nodes are also at distance $t$ from $f$. Hence, by using Proposition~\ref{prop:U_AF_repair}, $\Lambda_{\text{U}}^{\text{AF}}(d)$ is a nonincreasing function of $d$ for $d \in \{k,k+1,\dots,n-1\}$. 
	\end{proof}
As a final remark, note that for codes of sufficiently high rate, the repair degree $d^\ast_{\text{AF}}=n-1.$
It is not difficult to show that this conclusion is also true for $d^\ast_{\text{IP}}$ as the
set $J$ does not change with the change of the repair degree. This is formally justified by an 
appropriate modification of the proof of Theorem~\ref{lemma:threshold_k}. 

	\section{Intermediate Processing for Evaluation Codes}\label{sec:IPC}
For the optimal choice, all participating helpers contribute equally and the IP technique is simply that of a regular MSR codes. Furthermore, Theorem~\ref{thm:gen_IP_sat} shows that the IP technique in the nonuniform download setting for the codes in Construction~\ref{const:stacked} relies on the IP technique of its component codes, which are regular MSR codes. Therefore, in this section we consider IP repair for $F$-linear MSR codes, showing that repair
complexity with IP repair can be lower that relying on the AF strategy. We focus essentially on one family of
MSR codes, the product-matrix codes of \cite{Rashmi11} and its generalization in \cite{Duursma2020}. 
Other examples implementing IP repair are given in \cite{Patra2022NodeRO}.
	
	%
	%
	
	\vspace*{-.1in}	\subsection{Product-matrix (PM) codes} \label{sec:pm}
	As our first goal, we rewrite the IP repair of PM codes originally introduced in Part I, Sec.II.A to fit the evaluation code paradigm. We begin with a brief introduction to the original description of the PM framework.
	PM codes, constructed in \cite{Rashmi11}, form a family of MSR codes with parameters $[n,k,d=2(k-1),l=k-1,\beta=1,M=k(k-1)]$. 
	The data file $\cF$ consists of $M$ uniformly chosen symbols from a finite field $F$. These symbols are organized to form two symmetric matrices $S_1,S_2$ of order $k-1$, each consisting of $\binom{k}{2}$ independent symbols and hence accounting for a total of
	$M$ symbols. The encoding matrix $\Psi$ is taken to be an $n \times d$ matrix such that $\Psi = \begin{bmatrix}
		\Phi |\Lambda\Phi
	\end{bmatrix}$ where $\Phi$ is a $n \times (k-1)$ Vandermonde matrix with rows of the form 
	$\phi_i = (1,x_i,x_i^2,\dots,x_i^{l-1}), i=1,\dots,n$ and $\Lambda = \text{Diag}(x_1^l,x_2^l,\dots,x_n^l)$ is a diagonal matrix 
	where $x_1,\dots,x_n$ are distinct nonzero elements of $F$.
	The encoded message is defined as $C = \Psi (S_1|S_2)^\intercal$ and the $l$ symbols of row $i$ of $C$ are stored in node $i$. Thus the
	$i$th node stores the $l$-vector $\phi_iS_1+\lambda_i\phi_iS_2$.
	
	The node repair process goes as follows: assuming that node $f \in [n]$ has failed, and the helper nodes are $D \subseteq [n]\setminus\{f\}, |D|=d$, helper node $i \in D$ sends the symbol of $F$ found as $(\phi_iS_1+\lambda_i\phi_iS_2)\phi_f^\intercal$. Since the submatrix $\Psi_D$ formed of the rows of $\Psi$ indexed by $D$ is invertible, node $f$ can calculate $S_1\phi_f^\intercal$ and $S_2\phi_f^\intercal$ from which it can compute its contents as $\phi_fS_1+\lambda_f\phi_fS_2$.
	
	To phrase this differently, let $s_1(y,z)$ and $s_2(y,z)$ be two symmetric polynomials over ${F}$ of degree at most $k-2$ in each of the
	two variables (this means, for instance, that $s_1(y,z)=s_1(z,y)$). Because of the symmetry, the total number of independent coefficients is $M$, so $s_1,s_2$ can be used to represent $\cF.$
	Letting $x_1,\dots, x_n$ be distinct points of $F$, we let node $i$ store the $l$ coefficients of the polynomial 
	$g^{(i)}(z)=s_1(x_i,z)+x_i^{k-1}s_2(x_i,z)$ for all $i\in[n].$
	
	Using this description of the codes, the IP repair process of Part I can be phrased as follows.
	Let $f \in [n]$ be the failed node, let $D$ be the set of $d$ helpers, and let $A$ be a set of helper nodes of size at least $d-k+1 = k-1$. For $h \in D$ define the polynomial
	\begin{equation}\label{eq:Lagrange}
		l^{(h)}(z) = \sum_{j=0}^{d-1} l^h_j {z}^j:=\prod_{\stackrel{i \in D}{i \ne h}}\frac{z-a_i}{a_h-a_i}
	\end{equation}
	of degree at most $d-1$. Then the set $A$ transmits the $l$-dimensional vector
	\begin{equation}\label{eq:eqtn4}
		{\xi}(f,A):= \sum_{h \in A}g^{(h)}(a_f)
		\left[\begin{array}{l}
			l^h_0+a_f^{k-1}l^h_{k-1}\\
			l^h_1+a_f^{k-1}l^h_{k}\\
			\hspace*{.3in}\vdots\\
			l^h_{k-2}+a_f^{k-1}l^h_{2k-3}	
		\end{array}\right].
	\end{equation}
	We show that (i), the failed node can recover its value based on the vector $\xi(f,{D})$, and (ii), the intermediate nodes
	can save on the repair bandwidth by processing the received information. To show (i) we prove
	\begin{lemma}\label{lemma_pm}
		The content of the failed node $f$ coincides with the vector $\xi(f,D)$, i.e.,
		$$
		g^{(f)}(z)=\sum_{i=0}^{l-1}(\xi(f,D))_i\,{z}^i.
		$$
	\end{lemma}
	\begin{proof}
		Consider the polynomial $H(z) = s_1(a_f,z)+z^{k-1}s_2(a_f,z)$ and note that $\deg (H)\le 2k-3=d-1$. Thus if we write
		$H(z)=\sum_{j=0}^{d-1} g_j z^j,$ then the polynomial $g^{(f)}$ defined above can be written as
		\\[.05in]
		\centerline{$g^{(f)}(z)=\sum_{j=0}^{k-2}(g_j+a_f^{k-1}g_{k-1+j})z^j.$}\\[.05in] 
		Rephrasing, the contents of the node $f$ is
		$$
		(g_0+a_f^{k-1}g_{k-1},g_1+a_f^{k-1}g_k,\dots,g_{k-2}+a_f^{k-1}g_{2k-3})^\intercal.
		$$
		At the same time, using \eqref{eq:Lagrange} we can write $H(z)$ in the Lagrange form
		$
		H(z) = \sum_{h \in D}g^{(h)}(a_f)l^{(h)}(z).
		$
		The coefficient vector of this polynomial is nothing but $\xi(f,D)$.
	\end{proof}
	To show part (ii) we note that the polynomials $\{l_h(z)\}_{h\in D}$ do not depend on $\cF$ and can be computed at any node in the network. So what we care to receive from the helper nodes are the multipliers $\{g^{(h)}(a_f)\}_{h\in D}$. 
	Hence, for any set of helper nodes with $|A| < d-k+1$, it is gainful to send $\{g^{(h)}(a_f)\}_{h\in A}$ 
	rather than the vector $\xi(f,A),$ since the former requires fewer than $l$ transmissions. At the same time,  when
	$|A|\ge d-k+1,$ we can transmit the vector $\xi(f,A)$ of dimension $l$, meeting the bound \eqref{eq:uni_IP} and
	reproducing the result from Part I.
	
	Using multilinear algebra notation (more on it in the next section), we can rephrase the code description as follows. 
	The encoding is defined as a linear functional 
	$$
	\phi\in(F^2\otimes S^2F^{k-1})^\ast,
	$$
	where $S^2F^{k-1}$ is the second symmetric power (this is another way of saying that the encoding relies on evaluations of symmetric polynomials).
	Node $i$ stores a restriction of $\phi$ to $x_i\otimes y_i\otimes F^{k-1}$, where $x_i=[1,a_i^{k-1}], 
	y_i=[1, a_i, \dots , a_i^{{k-2}}].$ The contents of the failed node is a vector in the $l$-dimensional subspace 
	$(x_f\otimes {y_f}\otimes F^{k-1})^\ast$, and the IP procedure recovers the coordinates of this vector in stages that correspond to moving along the repair graph toward the failed node.
	A general version of this idea underlies the repair procedure in the following sections.

		\begin{example}\label{example:pm}
			{\rm	Consider the $[n\ge 7,k=4,d=6,l=3,\beta=1,M=12]$ PM MSR code, placed on the graph shown in Fig.~\ref{fig:graph_pm}. This graph
				should be thought of as a subgraph in a large storage network, formed by locating a helper set for the failed node.
				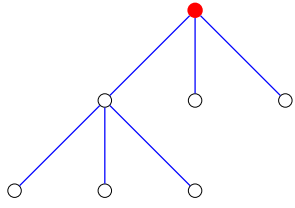
\begin{figure}[ht]	\begin{center}\scalebox{0.6}
						{\begin{tikzpicture}[roundnode/.style={circle, draw=black, inner sep=3pt},
								rootnode/.style={circle, draw=red, very thick,  inner sep=3pt,fill=red}			]
								\node[rootnode] (1) at (0,0) {};
								\node[roundnode] (2) at (-2,-2) {};
								\node[roundnode] (3) at (2,-2) {};
								\node[roundnode] (4) at (-4,-4) {};
								\node[roundnode] (5) at (-2,-4) {};
								\node[roundnode] (6) at (0,-4) {};
								\node[roundnode] (7) at (0,-2) {};
								\path[-] (1) edge [blue,thick] (2);
								\path[-] (1) edge [blue,thick] (3);
								\path[-] (2) edge [blue,thick] (4);
								\path[-] (2) edge [blue,thick] (5);
								\path[-] (2) edge [blue,thick] (6);
								\path[-] (1) edge [blue,thick] (7);
						\end{tikzpicture}}
					\end{center}
					\caption{The graph for Example~\ref{example:pm}}\label{fig:graph_pm}
				\end{figure}		
				Suppose that the root node is erased, and the remaining 6 nodes form the helper set. Since each of them contributes one symbol, the AF 
				repair procedure requires transmission of $9$ field symbols over the edges to complete the repair. In particular, the left most of the three vertices adjacent to $v_f$ sends 4 symbols over the edge connecting it to $v_f.$ At the same time, using the IP procedure described above, this node may only send $l=3$ symbols, showing that a total of $8$ transmissions are sufficient. This shows the bandwidth saving capabilities of the IP procedure.}
		\end{example}

	%
	%
	
	\subsection{Linear-algebraic notation}\label{sec:tensors}
	In this section, we introduce elements of notation used below to define code families for which we design IP procedures of node repair.
	
	For a linear space $U$ over $F$ we denote by $U^\ast$ its dual space; its elements are linear functionals of the form $\phi: U\to {F}$.
	The spaces $U$ and $U^{\ast}$ have the same dimension and $(U^{\ast})^{\ast}\cong U$. A {\em restriction} of $\phi$ to a subspace
	$V\subset U$ is denoted as $\phi \upharpoonright V$.
	
	Let $U,V$ be linear spaces of dimensions $m$ and $n$, respectively, and let us fix bases
	$\{\overline{u}_i\}_{i=1}^m$ and $\{\overline{v}_j\}_{j=1}^n.$ The tensor product of $U$ and $V$ is a linear space
	$
	U \otimes V=\{\sum_{ij}a_{ij}\overline{u}_i \otimes \overline{v}_j, a_{ij}\in {F}\}
	$
	where $a_{ij} \in {F}$ and the tensors $\overline{u}_i \otimes \overline{v}_j$ form a basis in $U\otimes V$ (thus $\dim (U\otimes V)=mn$).
	By definition, $u \otimes V=\{\sum_ja_j u \otimes \overline{v}_j, a_j\in {F}\}$ and $u \otimes V \subseteq U \otimes V$
	The dual of a tensor product is the tensor product of duals, i.e., $(U \otimes V)^{\ast} = U^{\ast}\otimes V^{\ast}$. 
	We denote by $T^pV := V^{\otimes p} $ the $p$-th tensor power of $V$. The dimension of $T^pV $ is $n^p$. 
	
	The {\em symmetric power} $S^pV$ is a linear space of symmetric tensors, i.e., the subspace of $T^pV$ formed of the tensors invariant under transformations of the form 
	$\overline{v}_1\otimes\dots\otimes\overline{v}_p\mapsto \overline{v}_{\sigma(1)}\otimes\dots\otimes\overline{v}_{\sigma(p)}$ for any
	permutation $\sigma.$ We write symmetric tensors as
	$$
	\sum_{\stackrel{i_1,i_2,\dots,i_p}{1 \le i_1\le i_2\le \dots\le i_p\le n}} a_{i_1i_2\dots i_p}\overline{v}_{i_1} \odot \overline{v}_{i_2} \odot\dots \odot \overline{v}_{i_p},
	$$
	where $\odot$ denotes the symmetric product and $a_{i_1i_2\dots i_p}$ are elements of $F$. 
	By definition, $\dim(S^pV)=\binom{n+p-1}{p}.$ The space $S^pV$ can be  thought of as a projection 
	$$S: T^pV \rightarrow S^pV$$
	that sends the tensor $\overline{v}_{i_1}\otimes \overline{v}_{i_2} \otimes \dots \otimes \overline{v}_{i_p}$ to $\overline{v}_{j_1} \odot \overline{v}_{j_2}\odot \dots \odot \overline{v}_{j_p}$ where $j_1\le j_2\le \dots \le j_p$ is a sorted copy of $i_1,i_2,\dots, i_p$. 
	
	Finally, $x\wedge y$ denotes the exterior (alternating) product of vectors, characterized by $x\wedge y=-y\wedge x$;
	hence $\overline{v}_{\sigma(1)}\wedge \overline{v}_{\sigma(2)}\wedge\dots\wedge \overline{v}_{\sigma(n)}=\sgn(\sigma)
	\overline{v}_1\wedge \overline{v}_2\wedge\dots\wedge \overline{v}_n$, where $\sgn(\sigma)$ is the signature of the permutation $\sigma$. The {\em exterior power} $\Lambda^pV$ is a vector subspace of dimension $\binom np$ spanned by elements of the form $\overline{v}_{i_1} \wedge \overline{v}_{i_2} \wedge\dots \wedge \overline{v}_{i_p}, 1 \le i_1< i_2< \dots< i_p\le n$, so a vector in $\Lambda^pV$ has the form
	$$
	\sum_{\stackrel{i_1,i_2,\dots,i_q}{1 \le i_1< i_2< \dots< i_q\le n}} a_{i_1i_2\dots i_q}\overline{v}_{i_1} \wedge \overline{v}_{i_2} \wedge\dots \wedge \overline{v}_{i_q}.
	$$
	The spaces $S^pV$ and $\Lambda^qV$ are formed by the action on $T^pV$ of the symmetric and alternating groups, respectively.

	By convention, $T^0V$, $S^0V$ and $\Lambda^0V$ are taken to be $F$.
	
	%
	%
	
	\subsection{Generalized PM codes}\label{sec:GPM}
	An extension of the PM construction was recently proposed in \cite{Duursma2020}.
	The construction of \cite[Sec.4]{Duursma2020} yields a family of MSR codes with parameters 
	$$
	n,k,d=\frac{(k-1)t}{t-1},l=\binom{k-1}{t-1}, M=t\binom kt, \quad 2\le t\le k\le n-1.
	$$
	In this section we follow the paradigm of evaluation codes to 
	introduce an IP node repair procedure for this code family.
	
	We start with a brief description of the code construction. Let $X={F}^t$ and $Y={F}^{k-t+1}$. Let $L:=X \otimes S^tY$ and note that
	$\dim(L)=M.$ The encoding $\phi:L\to F^{nl}$ is an $F$-linear map. To define a concrete encoding procedure, we fix a basis in $L^*$ and 
	let the coordinates of $\phi$ be the contents of the stored data. 
	
	To support the data reconstruction and node repair tasks, we further choose, for each $i \in [n],$ a pair of vectors $x_i\in X$ and $y_i\in Y$ 
	such that
	\begin{itemize}
		\item[(i)] Any $t$-subset of $x_i$'s spans $X$.
		\item[(ii)] Any $(k-t+1)$-subset of $y_i$'s spans $Y$.
		\item[(iii)] Any $d$ subspaces $x_i\otimes y_i \odot S^{t-2}Y$ span $X \otimes S^{t-1}Y.$
	\end{itemize}
	The first two properties enable data reconstruction, while the node repair property depends on the third condition \cite{Duursma2020}. 
	
	With these assumptions, the contents of node $i$ correspond to the restriction 
	$\phi \!\upharpoonright\! x_i \otimes y_i \odot S^{t-1}Y \in  (x_i \otimes y_i \odot S^{t-1}Y)^\ast.$ This is consistent with the
	code parameters: indeed, an element in $(x_i \otimes y_i \odot S^{t-1}Y)^\ast$ is completely described by its evaluations on a basis of the space $x_i \otimes y_i \odot S^{t-1}Y,$ which requires storing exactly $l=\binom{k-1}{t-1}$ evaluations.

	As before, let $f \in [n]$ be the (index of the) failed node and let $D \subseteq [n]\setminus\{f\}$ be the helper set. 
	Note that we wish to recover the restriction	$\phi \upharpoonright x_f \otimes y_f \odot S^{t-1}Y.$ 
	Choose a basis for $x_f \otimes y_f \odot S^{t-1}Y$ and let 
	$
	x_f \otimes y_f \odot (\overline{y}_{i_1}\odot \dots \odot \overline{y}
	_{i_{t-1}})
	$ 
	be one of the basis vectors. Let 
	$$
	\{\underline{y}_{j_1} \odot \dots \odot \underline{y}_{j_{t-2}}, 1\le j_1\le j_2\dots\le
	j_{t-2}\le k-t+1\}
	$$
	be a basis of $S^{t-2}Y$. 
	The helper node $i \in D$ transmits to the failed node the restriction of $\phi$ to the set of vectors 
	$\{x_i \otimes y_i \odot (\underline{y}_{j_1} \odot \dots \odot \underline{y}_{j_{t-2}})\odot y_f\}$.
	
	\vspace*{.1in} It becomes easier to think of the above construction once we connect it with PM codes described in Sec.~\ref{sec:pm}.
	For that, take $t=2.$ In this case, the file size is 
	$$
	\dim(L) = \dim(X \otimes S^2Y) = \dim(F^2 \otimes S^2F^{k-1}) = k(k-1).
	$$
	Node $i$ stores $\phi \upharpoonright (x_i \otimes y_i \odot Y)$, i.e, $\phi$ evaluated at a basis of $x_i \otimes y_i \odot Y,$ which requires 
	storing exactly $\dim(Y)=k-1$ symbols. Each node can calculate the symbol $\phi(x_i\otimes y_i \odot y_f)\in F$. Now notice that $d$ vectors $\{x_i \otimes y_i\}$ span $X \otimes Y,$ and so $d$ values $\phi(x_i\otimes y_i \odot y_f)$ account for the evaluations of $\phi$ on $X \otimes Y \odot y_f$. From this set of evaluations, we can calculate $\phi$ on $x_f \otimes Y \odot y_f$ which by the symmetric product property is the 
	same as $x_f \otimes y_f \odot Y$. These evaluations form the contents of the failed node.
	
	The IP repair for this construction works as follows. By (iii) above we can write
	\begin{align*}
		&x_f \otimes y_f \odot (\overline{y}_{i_1}\odot \dots \odot \overline{y}_{i_{t-1}})
		= x_f \otimes (\overline{y}_{i_1}\odot \dots \odot \overline{y}_{i_{t-1}}) \odot y_f\\
		&= \sum_{i \in D}\sum_{j_1,\dots,j_{t-2}}a_{i,j_1,\dots, j_{t-2}}x_i \otimes y_i \odot \cY_{j_1,\dots,j_{t-2}}\odot y_f,
	\end{align*}
	where we denoted $\cY_{j_1,\dots,j_{t-2}}=\underline{y}_{j_1} \odot \dots \odot \underline{y}_{j_{t-2}}.$ Again similarly to the PM codes, any set $A \subseteq D$ with $|A| \ge d-k+1$ can transmit the following single evaluation of $\phi$
	along the path to $f$:
	\begin{align*}
		&\phi\Big(\sum_{i \in A}\sum_{j_1,\dots,j_{t-2}}a_{i,j_1,\dots, j_{t-2}}x_i \otimes y_i \odot \cY_{j_1,\dots,j_{t-2}}\odot y_f\Big)\\
		&=\sum_{i \in A}\sum_{j_1,\dots,j_{t-2}}a_{i,j_1,\dots, j_{t-2}}\phi(x_i \otimes y_i \odot \cY_{j_1,\dots,j_{t-2}}\odot y_f).
	\end{align*}
	This can be done for all basis vectors of the chosen basis of $x_f \otimes y_f \odot S^{t-1}Y,$ and that requires $l=\binom{k-1}{t-1}$ transmissions, which matches the lower bound \eqref{eq:uni_IP}. Note that the AF repair would require any set $A$ of helpers to transmit $\beta |A|$ symbols of $F$, which is greater than $l$ for $|A|> d-k+1.$ 
	
	We have shown that IP repair can outperform direct relaying.  Let us give an example to support this claim.
	\begin{example}\label{example:gen_pm}
		{\rm	Consider the use of $[n=7,k=5,d=6,l=6,\beta=3,M=30]$ generalized PM codes for the graph shown in Fig.~\ref{fig:graph}.
			\begin{figure}[ht]	\begin{center}\scalebox{0.6}
					{\begin{tikzpicture}[roundnode/.style={circle, draw=black, inner sep=7pt},
							rootnode/.style={circle, draw=red, very thick,  inner sep=7pt,fill=pink}			]
							\node[rootnode] (1) at (0,0) {1};
							\node[roundnode] (2) at (-1.5,-3) {2};
							\node[roundnode] (3) at (1.5,-3) {3};
							\node[roundnode] (4) at (-3,-6) {4};
							\node[roundnode] (5) at (-1,-6) {5};
							\node[roundnode] (6) at (1,-6) {6};
							\node[roundnode] (7) at (3,-6) {7};
							\path[-] (1) edge [blue,thick] (2);
							\path[-] (1) edge [blue,thick] (3);
							\path[-] (2) edge [blue,thick] (4);
							\path[-] (2) edge [blue,thick] (5);
							\path[-] (3) edge [blue,thick] (6);
							\path[-] (3) edge [blue,thick] (7);
					\end{tikzpicture}}
				\end{center}
				\caption{The graph for Example~\ref{example:pm}}\label{fig:graph}
			\end{figure}
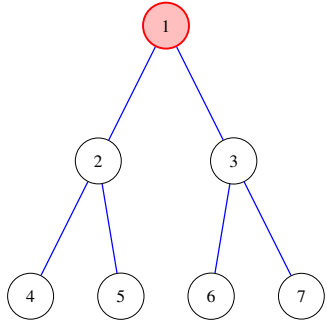 
			Suppose that $F=\ff_{16}$, and $X=F^3, Y=F^3$. Choose distinct $\{a_i\}_{i=1}^7$ from $F$ and let $x_i = \begin{bmatrix}
				1 & a_i^2 & a_i^6\end{bmatrix}, y_i = \begin{bmatrix}1 & a_i & a_i^3\end{bmatrix}$. The file of size 30 stored by the system is defined by the evaluation of a functional $\phi$ on $X \otimes S^3Y$. Node $i$ stores the restriction of $\phi$ to $x_i \otimes y_i \odot S^2Y$, hence storing 6 symbols per node. 
			
			For the repair of node 1, node $i$ sends the evaluations of $\phi$ restricted to the tensors $\{x_i \otimes y_i \odot \hat{y}_j \odot y_1\}_{j=1}^3$ where the set $\{\hat{y}_j\}$ forms a basis of $Y$. Since the six tensors $x_i \otimes y_i$ span $X \otimes Y$, the set $\{x_i \otimes y_i \odot \hat{y}_j \odot y_1, 1\le j\le 3; 2\le i\le 7\}$ spans $X\otimes Y \odot Y \odot y_1 \equiv X \otimes y_1 \odot Y \odot Y$. Hence 
			$$ \phi(x_1 \otimes y_1 \odot \cY_{j_1,j_2}) = \sum_{i=2}^7\sum_{j=1}^3b_{i,j,j_1,j_2}\phi(x_i \otimes y_i \odot \hat{y}_j \odot y_1)$$ where $\{\cY_{j_1,j_2}\}$ is a basis of $S^2Y$. 
			
			Consider nodes 2 and 3. Instead of relaying $\{\phi(x_i \otimes y_i \odot \hat{y}_j \odot y_1): i=2,4,5, j=1,2,3\},$ node 2 can transmit $\{\sum_{i=2,4,5}\sum_{j=1}^3b_{i,j,j_1,j_2}\phi(x_i\otimes y_2 \odot \hat{y}_j \odot y_1)\}$ for all $\{\cY_{j_1,j_2}\}$. The former requires 9 symbol transmissions while the later requires only 6, and the same holds true for node 3. In total, the AF repair procedure would require transmission of 
			$3\cdot(1+1+1+1+3+3)=30$ symbols, while the IP procedure requires only $3\cdot(1+1+1+1+2+2)=24$ symbol transmissions. This matches exactly with the bound in Corollary~\ref{cor:gen_IP_MSR}.}
	\end{example}

	\section{Error Correction during Repair with GRCs on Graphs}\label{sec:errors}
In this section, we extend the analysis of node repair using the GRC framework to the adversarial case, assuming that parts of the network (some helper nodes participating in the repair) contribute corrupted information. As explained in Section~\ref{sec:BRB}, such nodes can have a detrimental effect on the repair process due to the possible error amplification by IP. We prove a variant of the cutset bound of Theorem~\ref{theorem:gen_cutset} for this case and propose a code construction to harness the bandwidth saving capabilities of IP while also counteracting the effects of adversarial nodes. Our proposed solutions
make use of tools and ideas from network coding.

\subsection{Network Coding Preliminaries}
	In this subsection, we briefly describe the general problem of single source network coding and introduce the network singleton bound. A network $(\hat{G}, v_s, U, \cR)$ consists of a directed acyclic graph $\hat{G} = (\hat{V},\hat{E})$ with a single source node $v_s \in \hat{V}$, a set of destination nodes
	$U \subseteq \hat{V} \setminus \{v_s\}$, and a set of non-negative integers $\cR = \{R_{a,b}: (a,b) \in \hat{E}\}$ to denote the set of capacities of the edges (links) in the network. 
	A link with unit capacity transmits one symbol of $F$ per single use. To accommodate integer capacities $R_{a,b} >1$, 
	we simply add parallel unit-capacity edges between the nodes $a$ and $b$.
	An error is said to occur when the output of such a unit-capacity edge is different from the input. For a partition $(A,B)$ of the set $\hat{V}$, let 
	$\text{\rm cut}_{\hat{G}}(A,B) = \{(a,b)\in \hat{E}: a\in A, b \in B\}$ and let $c_{\hat{G}}(A,B) = \sum_{(a,b) \in \text{\rm cut}_{\hat{G}}(A,B)}R_{a,b}$. Additionally, for any two nodes $s$ and $u$, let 
	$$
	c_{\hat{G}}(s,u) = \min_{\stackrel{(S,U) \text{ is a partition of $\hat{V},$}}{s\in S, u \in U}}c_{\hat{G}}(S,U).
	$$

	\begin{definition}
		A network code over the code alphabet $\cX$ for the network $(\hat{G}, v_s, U, \cR)$ with source message set $\cZ$ is a family of local encoding functions $\{\phi_{(a,b)}: (a,b)\in \hat{E}\}$ such that $\phi_{(v_s,b)}: \cZ \rightarrow \cX^{r_{v_s,b}}$ for every $(v_s,b) \in \hat{E}$ and $\phi_{(a,b)}: \prod_{(c,a) \in \hat{E}}\cX^{r_{c,a}} \rightarrow \cX^{r_{a,b}}$ for every $(a,b) \in \hat{E}, a \ne v_s$, where $0 \le r_{a,b} \le R_{a,b}$. Such a code is said to correct $t$ errors if it recovers the source message at each of the destinations as long as at most $t$ \emph{unit capacity} links are subjected to errors.
	\end{definition}
	The network Singleton bound\cite{YeungCai2006} is as follows:
	\begin{lemma}\label{lemma:net_singleton}
		Let $(\hat{G},v_s,U,\cR)$ be an acyclic network and let $\hat{c}=\min_{u \in U} c_{\hat{G}}(s,u)$. 
		If there exists a $q$-ary $t$ error-correcting code for the network then the number of messages that can be transmitted from $s$ to $U$ 
		is at most $q^{\hat{c}-2t}.$
	\end{lemma}

	\subsection{A cutset bound for repair with adversarial nodes}\label{sec:bound_adv}

We now extend the above framework to node repair. Recall that helper node $h \in D$ of a GRC code contributes $\beta_{\tau(h)}$ symbols for the repair of failed node $f$ for some assignment $\tau: D \rightarrow [d]$.  The following definition relies on notation from Sec.~\ref{sec:IP-GRC}.
	\begin{definition}
For a failed node $f$ and its repair graph $G_{f,D} = (V_{f,D}, E_{f,D})$, assumed to be a tree, define a directed acyclic graph $\tilde{G}_{f,D} = (\tilde{V}_{f,D},\tilde{E}_{f,D})$ in the following way:
		\begin{enumerate}
			\item Let $\tilde{V}_{f,D} = V_{f,D}$ and for every edge in $E_{f.D}$ define an edge in $\tilde{E}_{f,D}$ whose direction is defined by the direction of the data flow in the repair process.
			\item Next, for every $h \in V_{f,D}$, $h \ne v_f$ add another vertex $\tilde{h}$ to $\tilde{V}_{f,D}$ and add a directed edge $(\tilde{h},h)$ to $\tilde{E}_{f,D}$.
		\end{enumerate}
	\end{definition}
	The purpose of adding additional nodes in the graph is to formally define the \emph{limited-power adversary} that we are considering. Recalling the notation introduced in Sec.~\ref{sec:general}, we assume that for a node $h \in V_{f,D}\setminus\{v_f\}$, $\tilde{h} \in \tilde{V}_{f,D}$ stores $W_h$ and computes the functions $\cG^{\tau}_{h,f}$. The repair data then is transmitted to $h \in \tilde{V}_{f,D}$ via the directed edge $(\tilde{h},h)$ which, after possibly receiving the repair data from other nodes, computes the function $I^{\tau}_{h,f}$ (denoted earlier in Sec.~\ref{sec:IP-GRC} by $I_{h,f}$ and forwards to the next node. The adversary can now be assumed to have control only over the nodes of the former type. The following relationship holds between the cuts of original repair graph $G_{f,D}$ and the modified graph $\tilde{G}_{f,D}$:
	\begin{proposition}\label{prop:cutset}
		For any set $A \subset D$ in $G_{f,D}$, 
		let $\tilde{A} = \{\tilde{h}: h \in A\}$ in $\tilde{G}_{f,D}$. Then
		$$
		\text{\rm cut}_{\tilde{G}}(A\cup \tilde{A},\tilde{V}_{f,D}\setminus (A\cup\tilde{A})) = \text{\rm cut}_{G}(A,V_{f,D}\setminus A).
		$$ 
	\end{proposition}
	\begin{proof}
		Clearly, $\text{\rm cut}_G(A,v_f\cup \{D\setminus A\}) \subseteq \text{\rm cut}_{\tilde{G}}(A\cup \tilde{A},\tilde{V}_{f,D}\setminus (A\cup\tilde{A}))$ as we have not deleted any of the edges involved in the repair. Further, every newly added edge has either both ends in $\tilde{A}$ or both ends outside of it, so none of them can be a part of $\text{\rm cut}_{\tilde{G}}(A\cup \tilde{A},\tilde{V}_{f,D}).$
	\end{proof}
Now we state our main result of this section which gives a variation of Theorem \ref{theorem:gen_cutset} for the adversarial case. Let 
$
\Omega_r(\cB) = \max_{R \subseteq [d],|R|=r} \sum_{i \in R}\beta_i
$ 
denote the sum of $r$ largest elements 
from $\cB$.
	\begin{theorem}\label{lemma:adv}
		Suppose that the data stored on the graph $G$ is encoded using an $[n,k,d,l,\cB = \{\beta_j\}_{j=1}^d,M]$ Regenerating code. For the repair of a failed node $f$, let $D$ be the set of chosen helper nodes. Suppose the limited power adversary has control over a set $T \subset D$, $|T| \le t$ of these helper nodes. For any subset $A \subseteq D$ of size at least $d-k+1+2t$ that contains $T$, we have
		$$c(A, V_{f,D}\setminus A) \ge M-\sum_{i=1}^{k-1}\min\{l,\Delta_{d-i+1}(\cB)\}+2\Omega_t(\cB).$$ 
	\end{theorem}
	\begin{proof} We begin with transforming the repair problem into a network coding problem $(\hat{G},v_s,U,\cR)$,
and then apply the network Singleton bound. For a selected repair protocol and transmission scheme, we set $\hat{G}$  to be the directed graph $\tilde{G}_{f,D}$, introduce a dummy source
node connected to all the vertices $\tilde{A}=\{\tilde{h}: h \in A\}$ of $\tilde{V}_{f,D}$ by 
infinite-capacity edges, and set the sink $U$ to be the failed node $v_f$. Fix some assignment $\tau$ and set the capacity of the 
edges to be equal to the number of symbols transmitted over them as determined by the repair protocol. Note that the capacity of 
 an edge of the form $(\tilde{h},h)$ for $h \in D$ is set to $\beta_{\tau(h)}$ by this assignment. Hence an 
 adversary controlling a set of $T$ nodes in $G_{f,D}$, can inject at most $\sum_{h \in T}\beta_{\tau(h)}$ errors. In 
 network coding terms, this implies that the adversary can cause errors in 
at most $\sum_{h \in T}\beta_{\tau(h)}$ unit-capacity edges in $\tilde{G}_{f,D}$.

We phrase the rest of the proof using the notation used to prove Theorem~\ref{theorem:gen_cutset}.
Recall that $R_A^f$ is the random variable which is a function of the contents of the helper node set $A$ in the original graph $G_{f,D}$ such that $H(W_f|R_A^f,S_{D\setminus A}^f)=0$. In words, $R_A^f$ is the random variable jointly produced from the data stored at the set $A$ of the graph $G_{f,D}$ that is to be sent to the failed node such that the repair process is successful. 
Switching to the language of network coding, the set of vertices $\tilde{A}\in\tilde V_{f,D}$ 
wants to communicate the message $R_A^f$ to the sink node $v_f$ via the network $\tilde{G}_{f,D}$, where at most $\sum_{h \in T}\beta_{\tau(h)}$ edges can inject errors. If the random variable $R_A^f$ is supported on a set  $\cZ$ and if $\hat{c}$ is the mincut between the source and the destination, then by  Lemma~\ref{lemma:net_singleton}, $	\log |\cZ| \le \hat{c}-2\sum_{h \in T}\beta_{\tau(h)}.$
Since $H(R_A^f) \le \log |\cZ|,$ we obtain
		$$
  H(R_A^f) \le \hat{c}-2\sum_{h \in T}\beta_{\tau(h)}.
		$$ 
		By Proposition~\ref{prop:cutset},
		$$c_{G}(A,V_{f,D}\setminus A) = c_{\tilde{G}}(A\cup \tilde{A},\tilde{V}_{f,D}\setminus (A\cup\tilde{A})) \ge \hat{c} \ge H(R_A^f) +2\sum_{h \in T}\beta_{\tau(h)}.$$
Since this is true for all choices of $\tau,$ the tightest lower bound is obtained when,
 for some $\tau$, $\sum_{h \in T}\beta_{\tau(h)} = \Omega_t(\cB)$. Now substituting the lower bound of Theorem~\ref{theorem:gen_cutset} concludes the proof.
	\end{proof}
	\begin{corollary}\label{cor:msr_adv}
		At the MSR point, the lower bound takes the form
		$$
  c(E, V_{f,D}\setminus E) \ge \Delta_{d-k+1}(\cB)+\Omega_t(\cB),
  $$
  which under the uniform download assumption further simplifies to
		$$ c(E, V_{f,D}\setminus E) \ge (d-k+1+2t)\beta = l+2t\beta.$$
	\end{corollary}
	\subsection{Code Construction}\label{sec:constructions}
	In this section, we present a code construction that combines error correction with IP, and analyze its performance relative to the bound of Theorem~\ref{lemma:adv}. The construction supports error control of systematic nodes in the encoding and is based upon concatenating rank-metric codes with GRC codes at the MSR point, i.e., Construction~\ref{const:stacked}. Such concatenation was previously used in \cite{Silberstein2015} for error correction during data recovery with full connectivity between the nodes. However, our purpose is to correct errors during the repair process itself, as was done in \cite{Ye16a}, without sacrificing the benefits of IP on graphs.
	
	\emph{Rank-metric codes:}
	Recall that 
	a rank metric code $\cC$ is an $F$-linear subspace of the space of matrices $F^{n \times m}$ 
	with distance $d(A,B)=\rank(A-B).$ By $d_{\text{\rm min}}(\cC)$ we denote the minimum distance of the code. 
	{For} our construction we use a classic family of MRD codes, namely the {\em Gabidulin codes}. To define them, recall that a linearized polynomial $f(x)\in \ff_{q^m}[x]$ of $q$-degree $t$ {is defined as}
	$$
	f(x) = \sum_{i=0}^ta_ix^{q^i}, \quad a_t \ne 0.
	$$
	
	\begin{definition} Let $N\le m$ be integers and let $g_1,\cdots,g_N$ be elements of $\ff_{q^m}$ linearly independent over $F$.
		An $[N,K,D=N-K+1]_{q^m}$ Gabidulin code maps a $K$-tuple $f_0,\dots,f_{K-1}$ of elements of $\ff_{q^m}$ to an $N$-tuple $\ffc = (f(g_1),f(g_2),\cdots,f(g_N)),$ where $f(x)=\sum_{i=0}^{K-1} f_ix^{q^i}$ is a linearized polynomial.
	\end{definition}

 \begin{tcolorbox}[width=\linewidth, sharp corners=all, colback=white!95!black]
	\noindent\emph{Code Construction:} The code construction is a concatenation of Gabidulin codes and MSR codes. Let $\cC_1$ be an
	$[N,K]_{q^m}$ Gabidulin code and let $\cC_2$ be an $F$-linear systematic GRC code at the MSR point with parameters $[n,k,d,l=N,\cB = \{\beta_j\}_{j=1}^d]$.  
	The data to be encoded comprises $Kkm$ message symbols of $F=\ff_q$, or equivalently $Kk$ symbols of $F_{q^m}$. Partition them in $k$ blocks of $K$ elements and encode each block using code $\cC_1$. This gives a matrix $A$ of dimensions $k \times N$ over $\ff_{q^m}$. 
	
	Next fix a basis of $\ff_{q^m}$ over $F$ and expand each row of the matrix $A$, obtaining $k$ matrices $B^{(i)},i=1,\dots,k$ over $F$. Let 
	$B^{(i)}_j$ be row $j=1,\dots,m$ of the matrix $B^{(i)}.$ For a fixed $j$ we form a $k\times N$ matrix $R_j=((B^{(1)}_j)^\intercal,\dots,(B^{(k)}_j)^\intercal)^\intercal.$ 
	
The code $\cC_2$ defines an $F$-linear encoding map $F^{kN}\stackrel{\cC_2}\longrightarrow F^{nN}$. We assume that the encoding is {\em systematic}, i.e., there are some $k$ nodes that contain the $kN$ data symbols. Encoding the matrices $R_j$ with the code $\cC_2$, we obtain $m$ codewords $C_j=\cC_2(R_j), j=1,\dots,m,$ viewed as $n\times N$ matrices over $F$. 
	
Finally, for a fixed $i=1,\dots,n$, we take the rows $(C_j)_i$
	and place them on the storage node $i.$ Thus, each node contains $m$ $N$-dimensional vectors over $F$, each of which is a coordinate 
	of its codeword of the MSR code $\cC_2.$

\end{tcolorbox}
 
	\begin{proposition}\label{prop:systematic}
		The contents of any of the $k$ systematic nodes, viewed as an $N$ dimensional vector over $\ff_{q^m},$ forms a codeword of $\cC_1$.
	\end{proposition}
	The proof follows from the fact that the encoding for these nodes is an identity map, so mapping their contents back to $\ff_{q^m}$
	recovers the original {codewords} of $\cC_1.$
	\begin{proposition}\label{prop:reg}
		The resulting code satisfies the $k$-node data reconstruction and $d$-node repair properties. 
	\end{proposition}
	The proof follows immediately from the properties of the MSR code $\cC_2.$
	\begin{remark}
		Note that for recovering a file of size $kKm$, the required download is $kNm$ symbols over $\ff$, i.e., the communication complexity increses by a factor of $\frac{1}{R_1}$ where $R_1= \frac{K}{N}$ is the rate of $\cC_1$. If the $k$ nodes contacted for data reconstruction are all the systematic nodes of $\cC_2$, then this extra download is not required, since each of the systematic nodes can locally decode its own Gabidulin codeword and send the decoded symbols, which requires $kKm$ symbol transmissions.
	\end{remark}
	
Next we show that the above code construction supports repair of systematic nodes with IP despite the presence of a limited number of adversarial nodes.
	\begin{theorem}
		Let $v_f$ be a systematic node. Suppose that a subset $T \subset D$ of helper nodes with $|T| \le t$ are controlled by a limited-power adversary. If $N-K \ge 2\Omega_t(\cB)$, then the repair procedure of $v_f$ recovers the original data.
	\end{theorem}
	\begin{proof}
		Recall that the adversary corrupts the data of at most $t$ nodes $\{W_h: h \in T\}$, but it does not affect the 
		evaluation of the functions $\{\cG^{\tau}_{h,f}\}$ and $\{I^{\tau}_{h,f}\}$ defined in Sec.~\ref{sec:general}. 
		Since the MSR code is defined over $F$, these functions are $F$-linear. W.l.o.g. suppose that the adversary changes 
		$W_h$ to $W_h+Z_h,$ where $Z_h \in \ff_{q^m}^l$. Node $h$ produces the faulty helper data 
		$$
		\cG^{\tau}_{h,f}(W_h+Z_h) = \cG^{\tau}_{h,f}(W_h)+\cG^{\tau}_{h,f}(Z_h) = S^{\tau}_{h,f}+\tilde{Z}^{\tau}_{h,f}
		$$
		where $\tilde{Z}_{h,f} \in \ff_{q^m}^{\beta_{\tau(h)}}$. Performing the IP transform $I_{h,f}$ on this data, we obtain
		$$
		I^{\tau}_{h,f}(S^{\tau}_{h,f}+\tilde{Z}^{\tau}_{h,f})=U^{\tau}_{h,f}S^{\tau}_{h,f}+U^{\tau}_{h,f}\tilde{Z}^{\tau}_{h,f}.
		$$
		Note that the rank of the error term $U^{\tau}_{h,f}\tilde{Z}^{\tau}_{h,f}$ is at most $\beta_{\tau(h)}$. Since all the nodes in the network 
		operate faithfully, oblivious to the data corruption in some of them, at the failed node the helper data has the form
		$$ 
		\sum_{h \in D}U^{\tau}_{h,f}S^{\tau}_{h,f}+ \sum_{h \in T}U^{\tau}_{h,f}\tilde{Z}^{\tau}_{h,f} = W_f+ \sum_{h \in T}U^{\tau}_{h,f}\tilde{Z}^{\tau}_{h,f}.
		$$ 
		Observe that	
		\begin{align*}
			\rank_q(\sum_{h \in T}U^{\tau}_{h,f}\tilde{Z}^{\tau}_{h,f}) \le \sum_{h \in T}\rank_q(U^{\tau}_{h,f}\tilde{Z}^{\tau}_{h,f})\le \sum_{h \in T}\beta_{\tau(h)} \le \Omega_t(\cB).
		\end{align*}
		By Proposition \ref{prop:systematic}, $W_f$ is a codeword of an $[N,K]$ Gabidulin code that can correct up to $\frac{N-K}{2}\ge \Omega_t(\cB)$ rank errors. Hence, the failed node will be able to correct the errors {and recover its value}.
	\end{proof}
	
	\begin{example}
Consider the graph with 10 nodes shown in Fig.~\ref{fig:repair_err} and assume that the node $v_f$ has failed, 
and all the other nodes form the helper set $D$. 
Suppose that one of the nodes, say node $3$, is adversarial, i.e., its contents have been arbitrarily altered. 
To choose a family of MSR codes, take the codes of \cite{Ye16a} with parameters $[n=10,k=5,d=7,l=15,\beta=5]$. Using AF 
repair, we contact $d+2t =9$ nodes and download $5$ symbols of helper data from each of them,
totaling 95 symbol transmissions. Using the code construction described above with an $[n=10,k=5,d=9,l=25,\beta=5]$ MSR code 
and a $[N=25,K=15,D=11]$ MRD code, we will be able to perform IP, correct one error, and achieve the communication complexity of 85 
symbols. Note that both codes store a file of size 75.
		
	\end{example}
	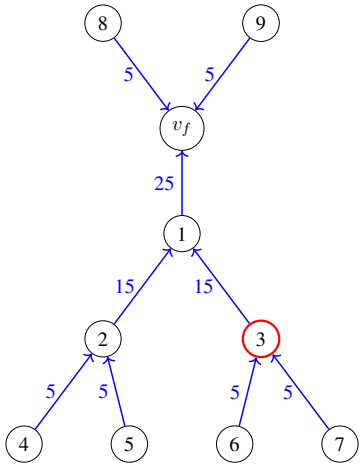
\begin{figure}[h]	
		\begin{center}\scalebox{0.7}
			{
				\begin{tikzpicture}[roundnode/.style={circle, draw=black},
					errornode/.style={circle, draw=red, very thick, minimum size=1mm}			]
					\node[roundnode] (0) at (0,2) {$v_f$};
					\node[roundnode] (1) at (0,0) {1};
					\node[roundnode] (8) at (-1.5,4) {8};
					\node[roundnode] (9) at (1.5,4) {9};
					\node[roundnode] (2) at (-1.5,-2) {2};
					\node[errornode] (3) at (1.5,-2) {3};
					\node[roundnode] (4) at (-3,-4) {4};
					\node[roundnode] (5) at (-1,-4) {5};
					\node[roundnode] (6) at (1,-4) {6};
					\node[roundnode] (7) at (3,-4) {7};
					\path[<-] (1) edge [blue,thick] node[style={anchor=east}] {15}(2);
					\path[<-] (1) edge [blue,thick] node[style={anchor=east}] {15}(3);
					\path[<-] (2) edge [blue,thick] node[style={anchor=east}] {5}(4);
					\path[<-] (2) edge [blue,thick] node[style={anchor=east}] {5}(5);
					\path[<-] (3) edge [blue,thick] node[style={anchor=east}] {5}(6);
					\path[<-] (3) edge [blue,thick] node[style={anchor=east}] {5}(7);
					\path[->] (1) edge [blue,thick] node[style={anchor=east}] {25}(0);
					\path[->] (8) edge [blue,thick] node[style={anchor=east}] {5}(0);
					\path[->] (9) edge [blue,thick] node[style={anchor=east}] {5}(0);
				\end{tikzpicture}
				
			}
			\caption{Repair with errors. Transmission from Node 1 to $v_f$ is brought down from 35 to 25 by using our code construction.}
			\label{fig:repair_err}
		\end{center}
	\end{figure}
	Although the code satisfies the $k$-node data recovery and $d$-node repair properties, it is not an MSR code because it does not meet the cutset bound with equality. Because of the \textit{two-layer} coding process, the rate of the resulting code is reduced from $\frac{k}{n}$ to $\frac{Kk}{Nn}$ and, as is to be expected, the rate decreases as we go for greater distance in the Gabidulin code, i.e., increase $N-K$.

	\section{Concluding remarks}
(1) The results of this paper suggest that design of codes for a distributed storage systems with graphical constraints is governed by the following principles. If we are free to adjust the repair degree in different rounds of repair, then we should use a code that supports {\em multiple repair degrees} such as the constructions in \cite{Ye16a}, \cite{liu2022optimal}, choosing the number of helpers that optimizes the repair bandwidth under the {\em uniform download} assumption. If the graph is sufficiently regular so that this degree does not depend on the failed node, we can instead rely on standard code families designed for a {\em fixed repair degree}. If the repair degree is fixed
by the system constraints, we should use the stacking construction of Section ~\ref{sec:stacking}, which allows to modify the download amounts based on the distance from the helper to the failed node in each round of repair. In other words, if a node is included in the set of helpers for a particular failed node and is close to this node, it provides a larger proportion of the repair data. If in the next round this node is again a helper, but is far in the graph from the (new) failed node, it provides a small amount of data or, possibly, no data at all.

The above solutions depend on the constraints of the storage systems. Whichever of these alternatives the
designer follows, it is also possible to incorporate IP repair with any of the mentioned code constructions,
gaining further savings of the repair complexity.
	
(2) In this work we analyzed in detail the problem of node repair on graphs, without saying much about the 
other basic task, namely, data retrieval. At the MSR point, the task of data retrieval from a regenerating 
code defined on an arbitrary graph becomes trivial. Indeed, since MSR codes are Maximum Distance Separable 
by definition, any set $A$ of $k$ or fewer nodes has to transmit $|A|\cdot l$ symbols and there is no hope 
of compressing this any further. This implies that in the restricted connectivity setting, when the Data 
Collector (DC) does not have direct access to $k$ nodes, standard relaying of data is optimal. The situation 
changes when we lift the MSR constraint. It is possible to estimate the minimum number of symbols required 
for data retrieval on graphs, and it is also possible to design a retrieval procedure with optimal complexity for MBR codes, for 
instance, for PM-MBR codes of \cite{Rashmi11}. These results are similar to the ones considered above, and 
will not be included here. 
	
(3) It is of interest to design an algebraic construction that supports the functionality of assigning
different download amounts to the same helper node in different rounds of repair. This would replace
the stacking construction, while possibly reducing the node size (subpacketization) of the coding scheme. We 
should note that we believe this is a difficult problem. It is also of interest to design codes that support 
repair in the presence of adversaries with optimal overhead for repair of all nodes rather than only 
systematic nodes.

	\appendices

	\section{Proof of Theorem \ref{theorem:opt_sol}:} \label{thm:nonuniform}
	Let $\cL$ be the solution space of $\{\beta_i\}$s such that $\beta_1 \le \beta_2 \le \cdots \le \beta_{n-k} = \beta_{n-k+1} = \cdots = \beta_{n-1}$. We know that the optimal solution of the LP lies in $\cL$. We call an assignment of $\{\beta_i\}$s satisfying the constraints of the problem and having $\beta_i =0$ for $i < n-d$ a $d$-repair scheme. Let $\cL_d$ be the set of feasible $d$-repair schemes, i.e.,
	\begin{equation}
		\begin{split}
			\cL_d = \{\{\beta_i\}: 
			0 = \beta_1=\cdots=\beta_{n-d-1} < \beta_{n-d} \\
			\le \beta_{n-d+1}\le \cdots \le\beta_{n-k} = \beta_{n-k+1} = \cdots = \beta_{n-1} \}.
		\end{split}
	\end{equation}
	We call an assignment of download amounts {\em uniform} if $\beta_i = \frac{l}{d-k+1}$ for all $i \ge n-d$ any other assignment with at least one $\beta_i \ne \beta:= \frac{l}{d-k+1}$ a {\em nonuniform} one.
	Clearly,
	$$
	\cL = \cup_{d=k}^{n-1}\cL_d.
	$$ 
	A priori, the uniform assignment is not necessarily a minimizer of the objective function
	in \eqref{eq:opt_prob}. Depending on the costs, there may exist nonuniform assignments that give a 
	lower value of the objective function than the uniform assignment in some $\cL_d$. The proof idea is 
	to show that for any such nonuniform assignment in $\cL_d$, there 
	exists a $d' < d$ for which the uniform assignment in $\cL_{d'}$ gives at most the same value of 
	the objective function as given by the nonuniform assignment in $\cL_d$. Let $\{\beta_i\}$ be 
	such a nonuniform assignment in $\cL_d$ such that 
 \begin{equation}\label{eq:beta}
		\sum_{i=n-d}^{n-1}b_i\beta_i < \sum_{i=n-d}^{n-1}b_i\beta.
	\end{equation}
	Since 
	\begin{equation}\label{eq:n-d}
		\sum_{i=1}^{n-k}\beta_i=\sum_{i=n-d}^{n-k}\beta_i \ge l,
	\end{equation} 
	and since the $\beta_i$'s
	are nondecreasing, for any nonuniform assignment, 
	\begin{equation}\label{eq:c}
		\beta_{n-k} = \beta+c
	\end{equation}
	for some integer $c, 0<c\le l-\beta$.
\setlength{\abovedisplayskip}{3pt}
\setlength{\belowdisplayskip}{3pt}	
 
 At the same time, there will be at least one $i \in \{n-d,\dots,n-k-1\}$ for which $\beta_i < \beta$. To see that inequality \eqref{eq:beta} can hold, rewrite it as
	$$
	\sum_{i=n-d}^{n-1}b_i(\beta_i - \beta) <0
	$$
	and observe that, since $\beta_i=\beta+c$ for all $i\ge n-k$, \eqref{eq:beta} implies that
	\begin{align*}
		\sum_{i=n-k}^{n-1}b_ic <\sum_{i=n-d}^{n-k-1}b_i(\beta - \beta_i) \le b_{n-d}c.
	\end{align*}
	Here the last inequality holds because the $b_i$'s are nonincreasing.
	So if $b_{n-d} > \sum_{i=n-k}^{n-1}b_i$, then a nonuniform assignment in $\cL_d$ is going to give a 
	lower value of the objective function than the uniform assignment in $\cL_d$. Now we show that,
	even if this is the case, we can lower the value of $d$ to some $d'$ such that, with uniform
	assignment in $\cL_{d'}$, the objective function does not increase.
	
	Below we will assume that $(\beta+c)|l$ (it is always possible to choose $l$ large enough
	so that it is divisible by the $\text{l.c.m.}(k+1,\dots,n-1)$).
	With this assumption, let us take $d' = \frac{l}{\beta+c}+k-1$ and consider the uniform assignment in $\cL_{d'}$. Letting $\beta' = \beta +c$, we have $l=(d'-k+1)\beta'$ and
	\begin{align*}
		\sum_{i=n-d}^{n-k-1}\beta_i \ge l-\beta' = \sum_{i=n-d'}^{n-k-1}\beta',
	\end{align*}
	Since the $b_i$'s are nonincreasing on $i$, the sum $\sum_{i=n-d}^{n-k-1}b_i\beta_i$ takes
	the smallest value when $\beta_i=0$ for $n-d\le i\le n-d'-1$ and $\beta_i=\beta'$ for
	$n-d'\le i\le n-1.$
	Hence, the uniform assignment in $\cL_{d'}$ results in communication cost at most equal to
	the cost of the initial nonuniform assignment in $\cL_d$ that we started with. Since this 
	argument applies for any $d$, there always exists a minimizing 
	solution to Problem \ref{eq:opt_prob} in the form of Eq. (\ref{eq:opt_sol}) for some $d$.

\bibliographystyle{IEEEtranS}
\bibliography{references}
\end{document}